\documentclass[5p]{elsarticle}
\journal{Computers \& Security} 

\usepackage[ruled]{algorithm2e} 

\SetAlFnt{\small}
\SetAlCapFnt{\small}
\SetAlCapNameFnt{\small}
\SetAlCapHSkip{0pt}
\IncMargin{-\parindent}

\usepackage{enumerate}
\usepackage{url}

\usepackage[usenames,dvipsnames]{xcolor}
\usepackage{colortbl}
\definecolor{Gray}{gray}{0.93}

\usepackage{tikz}
\usetikzlibrary{arrows.meta}

\usepackage{pgfplots}
\pgfplotsset{compat=1.12}

\usepackage[textsize=tiny, textwidth=1.3cm, colorinlistoftodos, disable]{todonotes}

\newcommand{\Aron}[1]{\todo[color=orange!40, linecolor=black!30!orange!60!white]{\textbf{Aron}: #1}}

\usepackage{amsmath}
\usepackage{amssymb}

\newdefinition{definition}{Definition}
\newtheorem{theorem}{Theorem}
\newtheorem{proposition}{Proposition}
\newproof{proof}{Proof}

\DeclareMathOperator*{\argmax}{argmax}
\DeclareMathOperator*{\argmin}{argmin}

\newcommand{\vect}[1]{\ensuremath{\boldsymbol{#1}}}

\newcommand{\calC}[0]{\ensuremath{\mathcal{C}}}
\newcommand{\calG}[0]{\ensuremath{\mathcal{G}}}
\newcommand{\calI}[0]{\ensuremath{\mathcal{I}}}
\newcommand{\calL}[0]{\ensuremath{\mathcal{L}}}

\newcommand{\vA}[0]{\vect{A}}
\newcommand{\vD}[0]{\vect{D}}
\newcommand{\vM}[0]{\vect{M}}

\newcommand{\vp}[0]{\vect{p}}

\newenvironment{revision}{}{}

\begin{document}
\setlength{\marginparwidth}{1.2cm} 

\title{Detection and Mitigation of Attacks on Transportation Networks as a\\Multi-Stage Security Game}  

\author[1]{Aron Laszka\corref{cor1}}
\address[1]{University of Houston, Houston, USA}
\author[2]{Waseem Abbas}
\address[2]{Information Technology University, Lahore, Pakistan}
\author[3]{Yevgeniy Vorobeychik}
\address[3]{Washington University in St. Louis, St. Louis, USA}
\author[4]{Xenofon~Koutsoukos}
\address[4]{Vanderbilt University, Nashville, USA}

\cortext[cor1]{Corresponding author}

\begin{abstract}
In recent years, state-of-the-art traffic-control devices have evolved from standalone hardware to networked smart devices.
Smart traffic control enables operators to decrease traffic congestion and environmental impact by acquiring real-time traffic data and changing traffic signals from fixed to adaptive schedules.
However, these capabilities have inadvertently exposed traffic control to a wide range of cyber-attacks, which adversaries can easily mount through wireless networks or even through the Internet.
Indeed, recent studies have found that a large number of traffic signals that are deployed in practice suffer from exploitable vulnerabilities, which adversaries may use to take control of the devices.
Thanks to the hardware-based failsafes that most devices employ, adversaries cannot cause traffic accidents directly by setting compromised signals to dangerous configurations.
Nonetheless, an adversary could cause disastrous traffic congestion by changing the schedule of compromised traffic signals, thereby effectively crippling the transportation network.
To provide theoretical foundations for the protection of transportation networks from these attacks,
we introduce a game-theoretic model of launching, detecting, and mitigating attacks that tamper with traffic-signal schedules.
We show that finding optimal strategies is a computationally challenging problem, and we propose efficient heuristic algorithms for finding near optimal strategies.
We also introduce a Gaussian-process based anomaly detector, which can alert operators to ongoing attacks.
Finally, we evaluate our algorithms and the proposed detector using numerical experiments based on the SUMO traffic simulator.
\end{abstract}


\maketitle

\section{Introduction}
\label{sec:intro}

The evolution of traffic signals from standalone hardware devices to complex networked systems has provided society with many benefits, such as reducing wasted time and environmental impact.
However, it has also exposed traffic signals to a variety of cyber-attacks.
While traditional hardware systems were susceptible only to attacks based on direct physical access, modern systems are vulnerable to attacks through wireless interfaces or even to remote attacks through the Internet.
To assess the severity of these threats in practice, Ghena et al.\ recently analyzed the security of real-world traffic infrastructure in cooperation with a road agency located in Michigan~\cite{ghena2014green}.
This agency operates around a hundred traffic signals, which are all part of the same wireless network, but the signals at every intersection operate independently of the other intersections.
The study found three major weaknesses in the traffic infrastructure: lack of encryption for the wireless network, lack of secure authentication due to the use of default usernames and passwords on the devices, and the presence of exploitable software vulnerabilities.

While all of these known weaknesses could be eliminated, it is extremely difficult to ensure that devices are free of any unknown weaknesses.
In general, it is virtually impossible---or prohibitively expensive---to ensure that a system is perfectly secure.
In addition to the general difficulty of attaining perfect security, traffic-control devices pose further challenges.
Similar to other distributed cyber-physical systems,  traffic-control systems have large attack surfaces, and they often have long system lifetime and complicated software-upgrade procedures, which makes fixing vulnerabilities difficult.
Consequently, operators cannot hope to stop all cyber-attacks since a determined and sophisticated attacker might always find a way to compromise some of the devices.
Therefore, instead of focusing solely on the first line of defense, operators must also consider minimizing the impact of successful cyber-attacks.


Due to hardware-based failsafes, compromising a traffic signal does not allow an attacker to set the signal to an unsafe configuration that would lead to traffic accidents, such as giving green light to two intersecting directions.
However, compromising a signal does enable tampering with its schedule, which allows the attacker to cause disastrous traffic congestion.
Such malicious cyber-attacks may be launched by any adversary whose interest is to case disruption and damage, ranging from cyber-terrorists to disgruntled ex-employees.
For instance, during the Los Angeles traffic engineers' 2006 strike, two disgruntled employees allegedly penetrated the traffic-control system of the city, and reprogrammed the traffic lights of four intersections to cause congestion: ``[t]he red signal would be on too long for the critical approach and the green signal would be on too long for the noncritical approach, thus resulting in long backups into the airport and other key intersections around the city''~\cite{bernstein2007key}.
Furthermore, terrorists could also mount these attacks in conjunction with physical attacks, thereby increasing their impact (e.g., delaying ambulances and firefighters).

To minimize the impact of attacks tampering with traffic signals, operators must be able to detect and mitigate them promptly and effectively.
In practice, the detection of novel cyber-attacks poses multiple challenges.
Since signature-based detectors are ineffective against novel attack, operators must employ anomaly-based detectors.
However, these detectors are prone to raising false alarms, which must be investigated manually, resulting in a waste of manpower and resources.
Considering the relative scarcity of attacks, the cost of these investigations may exceed the benefit of early attack detection and mitigation.
Therefore, when configuring the sensitivity of detectors, operators must carefully balance the cost of false alarms and the risk from delayed detection.
Moreover, sophisticated attackers can act strategically by mounting stealthy attacks, which delay detection but still cause significant impact.
In light of this, operators must also plan their defense strategically, by anticipating the attackers' responses.


\subsection*{Contributions}

In~\cite{laszka2016vulnerability}, we introduced an approach for evaluating the vulnerability of transportation networks to cyber-attacks that tamper with traffic-control devices.
In this paper, we extend this approach by considering detectors and countermeasures that operators can implement to mitigate these attacks. 
In particular, we introduce a game-theoretic model, in which an operator can setup anomaly-based detectors and mitigate ongoing attacks by reconfiguring traffic control.
Similar to~\cite{laszka2016vulnerability}, we build on the cell-transmission model introduced by Daganzo~\cite{daganzo1994cell,daganzo1995cell}.
\begin{revision}
To the best of our knowledge, our work is the first to consider the problem of designing and deploying systems based on traffic-sensors measurements to detect tampering attacks against traffic control. \end{revision} Our main contributions in this paper are the following:
\begin{itemize}
\item We formulate a multi-stage security game that models the detection and mitigation of cyber-attacks against transportation networks.
\item We propose an efficient metaheuristic search algorithm for finding detector configurations that minimize losses in the face of strategic attacks.
\item We introduce an anomaly-based detector for attacks against traffic control, which is built on a Gaussian-process based model of normal traffic.
\item We evaluate our detector and algorithms based on detailed simulations of traffic using SUMO.
\end{itemize}

\subsection*{Outline}
The remainder of this paper is organized as follows.
In Section~\ref{sec:model}, we introduce our game-theoretic model of detecting and mitigating attacks against transportation networks.
In Section~\ref{sec:analysis}, we present computational results on our model and propose efficient heuristic algorithms.
In Section~\ref{sec:detector}, we introduce a Gaussian-process based detector for attacks against traffic control.
In Section~\ref{sec:numerical}, we use detailed simulations of transportation networks to evaluate our detector and the heuristic algorithms.
In Section~\ref{sec:related}, we discuss related work on the vulnerability of transportation networks, configuration of attack detectors, and game theory for security of cyber-physical systems.
Finally, in Section~\ref{sec:concl}, we offer concluding remarks.

\section{Game-Theoretic Model of Attacks on Traffic Signals}
\label{sec:model}

\begin{revision}
In this section, we introduce our model of launching, detecting, and mitigating cyber-attacks against traffic control in transportation networks.
Our model includes two agents: an attacker who can launch cyber-attacks and a defender who attempts to detect and mitigate them.
Since these agents may anticipate and react to each other's actions, we formulate our model using game theory, which enables us to capture the strategic interactions between the two agents.
For a list of symbols used in our model, see Table~\ref{tab:symbols_game}.

\begin{table}[h!]
\begin{revision}
\caption{List of Symbols}
\label{tab:symbols_game}
\centering
\renewcommand*{\arraystretch}{1.0}
\setlength{\tabcolsep}{3pt}
\begin{tabular}{| c | p{6.5cm} |}
\hline
Symbol & \multicolumn{1}{l|}{Description} \\
\hline
\multicolumn{2}{|c|}{Traffic Model} \\
\hline
$Q_i$ & maximum number of vehicles that can flow into or out of cell $i$ \\
\rowcolor{Gray} $\delta_i$ & ratio between free-flow speed and backward propagation speed of cell $i$ \\
$N_i$ & maximum number of vehicles in cell $i$ \\
\rowcolor{Gray} $d_i$ & demand (inflow) at source cell $i$ \\
$\Gamma^{-1}(i)$ & set of predecessor cells to cell $i$ \\
\rowcolor{Gray} $p_{ki}^t$ & inflow proportion from cell $k$ to signalized intersection $i$ \\
$\calI$ & set of signalized intersections \\
\hline
\multicolumn{2}{|c|}{Game Constants and Functions} \\
\hline
$\calI_D$ & set of signalized intersections with attack detectors \\
\rowcolor{Gray} $B$ & attacker's budget for compromising traffic signals \\
$\Delta_D(\vD, \vA)$ & detection delay for detector configuration $\vD$ and attack $\vA$ \\
\rowcolor{Gray} $\Delta_M$ & length of mitigation (before returning to normal operation) \\
$T$ & congestion with default traffic control \\
\rowcolor{Gray} $T_A(\vA)$ & congestion as a result of attack $\vA$ (before mitigation) \\
$C$ & cost of investigating a false alarm \\
\rowcolor{Gray} $T_M(\vA, \vM)$ & congestion after mitigation $\vM$ of attack $\vA$ \\
$\calG(\vD, \vA, \vM)$ & attacker's gain for actions $(\vD, \vA, \vM)$ \\
\rowcolor{Gray} $\calL(\vD, \vA, \vM)$ & defender's loss for actions $(\vD, \vA, \vM)$ \\
\hline
\multicolumn{2}{|c|}{Game Variables} \\
\hline
$\vD$ & defender's detector configuration action (Stage I) \\
\rowcolor{Gray} $D_i$ & false-positive rate of detector at intersection $i \in \calI_D$ \\
$\vA = (\calI_A, \hat{\vp})$ & attacker's attack action (Stage II) \\
\rowcolor{Gray} $\calI_A$ & set of signals compromised by the attacker \\
$\hat{p}_{ki}$ & modified (by attack or mitigation) inflow proportion from cell $k$ to intersection $i$ \\
\rowcolor{Gray} $\vM$ & defender's mitigation action (Stage III) \\
\hline
\end{tabular}
\end{revision}
\end{table}

\subsection{Traffic Model}
\label{sec:traffic_model}

First, we introduce Daganzo's cell transmission model, the traffic model on which our game-theoretic model, our analysis, and our numerical evaluation are built. 
Here, we provide only a very brief summary of this traffic model, focusing on the notation that will be used throughout the paper.
For a detailed description of the model, we refer the reader to~\cite{daganzo1994cell,daganzo1995cell,ziliaskopoulos2000linear}.
\footnote{For readers who are familiar with the cell-transmission model, we recommend to continue with Section~\ref{sec:game_model}.}

The cell transmission model divides a road network into \emph{cells}, which represent homogeneous road segments, and divides time into uniform \emph{intervals}.
The length of a road segment corresponding to a cell is equal to the distance traveled in light traffic by a typical vehicle in one time interval.
Each cell $i$ has three sets of parameters:
$N_i^t$ is the maximum number of vehicles that can be present in cell $i$ at time $t$;
$Q_i^t$ is the maximum number of vehicles that can flow into or out of cell $i$ during time interval~$t$;
and $\delta_i^t$ is the ratio between the free-flow speed and the backward propagation speed of cell $i$ at time $t$ (see~\cite{ziliaskopoulos2000linear} for a detailed explanation).%
\footnote{This constant is used to quantify how the speed of traffic decreases as the cell becomes congested, and can model traffic phenomena such as shockwaves.} 
Cells that model road segments where vehicles can enter traffic are called \emph{source cells}, and each source cell $i$ has a traffic demand parameter $d_i^t$, which is the number of vehicles entering traffic at cell $i$ in time interval $t$.
Cells where vehicles may exit traffic are called \emph{sink cells}.

Every cell is connected to one or more other cells: cells that correspond to consecutive road segments or road segments that are joined by an intersection are connected.
The set of cells from which vehicles can move into cell $i$ is called the set of \emph{predecessor cells}, denoted by $\Gamma^{-1}(i)$.
Cell that have multiple predecessors are called \emph{merging cells}, while cells that are the predecessors of multiple cells are called \emph{sink cells}.
To model signal control at intersections, we follow Daganzo's proposition~\cite{daganzo1995cell} and introduce the time-dependent parameters $p_{ki}^t$ controlling the inflow proportions of merging cell $i$.
We let $\calI$ denote the set of merging cells that model signalized intersections.

To solve the traffic model (i.e., to determine the traffic flow for a given network and set of parameters), we use Ziliaskopoulos's linear programming approach~\cite{ziliaskopoulos2000linear}.
The objective of this linear program is the sum of the number of vehicles traveling (i.e., number of vehicles on the road) over time, which is clearly equal to the \emph{total travel time} of all the vehicles.
As a consequence, we can use the value of the linear program---which can be computed efficiently for a given instance---as a measure of network congestion.
\end{revision}

Finally, we consider relatively short-term attack scenarios, in which the parameters of the cells and the default (i.e., unattacked) schedules of the traffic signals are constant.
Hence, in our game-theoretic model, we will omit the superscript $^t$ from $Q_i^t$, $N_i^t$, $\delta_i^t$, and~$p_{ki}^t$.

\subsection{Multi-Stage Security Game}
\label{sec:game_model}

We model defensive countermeasures and attacks in a transportation network as a two-player multi-stage security game between a defender and an attacker~\cite{laszka2019towards}.
The defender represents the operator of the transportation network, who can configure traffic-control devices and aims to minimize congestion in the network.
The attacker represents any strategic adversary that can compromise and tamper with traffic signals and aims to maximize congestion.


In a nutshell, our game consists of the following three stages.
\begin{enumerate}[I.]
\item \textbf{Detector Configuration}: 
In the first stage, the defender configures detectors, which are deployed in the transportation network, to detect cyber-attacks against traffic control.
The detectors may be traffic-anomaly based detectors or conventional cyber-security intrusion detection systems (IDS).
When configuring these detectors, the defender should anticipate the attacker's possible adversarial actions in the second~stage.
\item \textbf{Attack on Traffic Control}:
In the second stage, the attacker mounts a cyber-attack against the transportation network by compromising traffic signals and tampering with their schedule to cause congestion.
When choosing its attack, the attacker must take into account both the detector configuration chosen by the defender in the first stage as well as the defender's possible mitigation actions in the third stage.
\item \textbf{Mitigation of Attack}:
In the third stage, the defender attempts to mitigate the attack by changing the configuration of uncompromised traffic-control devices to minimize congestion. 
As this is the final stage, mitigation simply needs to respond to the particular attack that was launched in the second stage.
\begin{revision}
Note that once the defender detects an ongoing attack, it should also try to regain control of the compromised devices as soon as possible.
However, since the devices may be physically scattered throughout the transportation network, regaining control of them can take a long time.
For instance, an attacker could have changed remote login passwords, severed commu\-nication-network connections, etc., forcing the defender to physically reset or reinstall compromised devices. 
Meanwhile, disastrous traffic congestions may form in the transportation network, which the defender must mitigate immediately.
\end{revision}
\end{enumerate}




\subsubsection{Stages and Strategic Choices}

Next, we provide a detailed description of the three stages of the game and the players' action spaces.

\paragraph{Stage I: Detector Configuration}

To detect stealthy cyber-attacks, detectors are deployed on the traffic-control devices at a subset $\calI_D$ of the signalized intersections~$\calI$.
These detectors can be either traffic-anomaly based detectors, such as the one that we will introduce in Section~\ref{sec:detector}, or conventional cyber-security intrusion detection systems. 
We assume that detectors are imperfect, which means that they may raise false alarms (when there is no attack in progress) and they may detect actual attacks with some delay.
The rate of false-positive errors and detection delay both depend on how sensitive a detector is: a more sensitive detector is more likely to raise false alarms but detects actual attacks earlier, and vice versa.
We assume that the operator can configure the sensitivity of every one of the $|\calI_D|$ detectors individually.


Specifically, in the first stage of the game, the defender chooses a sensitivity configuration~$\vD$ for the detectors. 
For ease of presentation, we let the sensitivity of the detector at each intersection $i \in \calI_D$ be represented by the false-positive rate of the detector.
Formally, a detector configuration~$\vD$ is an $|\calI_D|$-dimensional non-negative vector, where $D_i$ is the rate of false alarms raised by the detector at intersection $i \in \calI_D$.

\paragraph{Stage II: Cyber-Attack on Traffic Control}

\Aron{Describe attacker model in more detail and motivate it!}
In the second stage of the game, the attacker compromises a subset of the traffic signals and changes their schedule.
We let $\calI_A \subseteq \calI$ denote the set of traffic signals that the attacker chooses to compromise.
We assume that the attacker is resource bounded, which means that it can compromise signals in at most $B \leq |\calI|$ intersections at the same time.
Hence, the attacker's choice $\calI_A$ must satisfy $|\calI_A| \leq B$.


Once the attacker has compromised a set of traffic signals~$\calI_A$, it can reconfigure every one of them. 
\begin{revision}
However, most traffic control devices have hardware-level safety mechanisms in practice (see, e.g.,~\cite{ghena2014green}), which constrain the configurations that may be set by an adversary.
In particular, traffic signals typically use malfunction management units as a safety feature against controller faults (e.g., overriding a faulty controller that would give green lights to two intersecting directions).
These hardware-based failsafes also limit the impact of cyber attacks (e.g., preventing the attacker from causing traffic accidents) by overriding invalid configurations.
In our traffic model, the attacker's reconfiguration corresponds to setting new inflow proportions $\hat{p}_{ki}$ for the cells in $\calI_A$.
Therefore, we can model hardware-level failsafes by requiring the inflow proportions chosen by the attacker to constitute a valid configuration.
\end{revision}
Specifically, we assume that the inflow proportions set by a feasible attack must sum up to~$1$ for each compromised intersection:
\begin{equation}
\forall i \in \calI_A:  \sum_{k \in \Gamma^{-1}(i)} \hat{p}_{ki} = 1 .
\end{equation}

In sum, we can represent a feasible attack action $\vA$ as a pair $\vA = (\calI_A, \hat{\vp})$ that satisfies $|\calI_A| \leq B$ and $\sum_{k \in \Gamma^{-1}(i)} \hat{p}_{ki} = 1$ for every $i \in \calI_A$,
where $\calI_A$ is the set of compromised signals, and $\hat{\vp}$ are the tampered signal schedules.

\paragraph{Between Stages II and III: Detection}

Once the attacker has perpetrated its attack in the second stage, the compromised signals begin to operate with tampered schedules, which results in increased congestion in the transportation network.
Eventually, the detectors deployed by the defender will detect the attack (based on either traffic or cyber anomalies).
We let $\Delta_D$ denote the detection delay, that is, the amount of time between the launch and detection of the attack.
The detection delay depends on both the configuration $\vD$ chosen by the defender and the attack $\vA$ chosen by the attacker, which we express by representing delay as a function $\Delta_D(\vD, \vA)$ of~$\vD$ and~$\vA$.
Once the attack is detected, the game progresses to the third stage.

\paragraph{Stage III: Mitigation of Attack}

In the third stage, the defender mitigates the detected attack by reconfiguring traffic-control devices to alleviate congestion.
We assume that the defender can reconfigure any device that is still under its control, that is, any traffic signal that is not compromised by the attacker.
Since the attacker has compromised signals $\calI_A$, the defender can set new inflow proportions $\hat{p}_{ki}$ for the cells $i$ in $\calI \setminus \calI_A$.
We again require the new configuration to be valid, which means that the inflow proportions must sum up to $1$ for each reconfigured intersection:
\begin{equation}
\forall i \in \calI \setminus \calI_A:  \sum_{k \in \Gamma^{-1}(i)} \hat{p}_{ki} = 1 .
\end{equation}
Finally, we let $\vM = \{ \hat{p}_{ki} \,|\, i \in \calI \setminus \calI_A\}$ denote the defender's mitigation (i.e., reconfiguration) action.

\subsubsection{Player's Utilities}

We now define the defender's and the attacker's utilities resulting from the various strategic choices that they can make in the game.
First, we let $T$ denote the level of congestion in the transportation network with the default configuration of traffic-control devices (i.e., with inflow proportions $\vp$).
In practice, we can quantify congestion $T$ as, e.g., the average or total travel time of the vehicles in the transportation network between their origin and destination.
Recall from Section~\ref{sec:traffic_model} that we can efficiently compute travel time with default proportions~$\vp$ in our traffic model using a linear program.

Next, we let $T_A(\vA)$ denote the level of congestion after the attack but before the mitigation, which depends on the attacker's action $\vA$ chosen in the second stage.
Similar to $T$, we can compute $T_A(\vA)$ using our traffic model with the default proportion $p_{ki}$ for every cell $i \in \calI \setminus \calI_A$ but with the adversarial proportion~$\hat{p}_{ki}$ for every cell $i \in \calI_A$.
Finally, we let $T_M(\vA, \vM)$ denote the level of congestion after the attack has been mitigated, which depends on both attacker's action $\vA$ and the defender's mitigation action $\vM$.

We let the attacker's gain $\calG$ (i.e., positive utility) for actions $(\vD, \vA, \vM)$ be the total impact of the attack in terms of increased congestion level:
\begin{align}
\calG(&\vD, \vA, \vM) \nonumber \\ &= \left(T_A(\vA) - T\right) \cdot \Delta_D(\vD, \vA) + \left(T_M(\vA, \vM) - T\right) \cdot \Delta_M ,
\end{align}
where $\Delta_M$ is the amount of time between mitigation and the transportation network returning to normal operation (e.g., manually resetting compromised devices).
The first term quantifies the impact of the attack before mitigation, while the second term quantifies impact after mitigation but before returning to normal operation.

Next, we define the defender's loss (i.e., negative utility) resulting from actions $(\vD, \vA, \vM)$.
Recall that the detectors deployed in the transportation network are imperfect, and each detector $i \in \calI_D$ is continuously generating false alerts (i.e., false-positive errors) at rate $D_i$.
Since the defender cannot tell which alerts are false, it has to investigate every single alert, which costs manpower and resources.
Hence, we define the defender's loss considering both the total impact of attacks and the cost of investigating false alerts:
\begin{equation}
\calL(\vD, \vA, \vM) = \calG(\vD, \vA, \vM) + \sum_{i \in \calI_D} D_i \cdot C ,
\end{equation}
where $C$ is the cost of investigating an alert.
The first term quantifies the total impact of the attack, while the second term captures the cost of investigating false alerts.

\subsection{Solution Concept and Problem Formulation}

We assume that both players have perfect information: in the second stage, the attacker knows the detector configuration $\vD$ chosen by the defender in the first stage; and in the third stage, the defender knows the attack action $\vA$ chosen by the attacker in the second stage.
We assume that the attacker has perfect information because we are considering a sophisticated, worst-case attacker, who has extensive knowledge of its target (i.e., Kerckhoffs's principle) and may know the algorithms or techniques employed by the defender for configuring the detectors.
On the other hand, we assume that the defender has perfect information because we are considering a smart transportation network with monitoring capabilities; however, this assumption could be relaxed.

\Aron{Check and revise if necessary!}
Under this assumption, we can model the players' optimal choices most naturally using the solution concept of \emph{subgame perfect equilibrium}.
Our goal is to find an optimal strategy for the defender by solving the game.
We can do so by finding the players' best-response actions for every stage using backward induction, that is, by solving each stage starting with the third and finishing with the first, in each stage building on the solutions for the subgame formed by the subsequent stages.
Given detector configuration $\vD$ and attack $\vA$, a best-response mitigation~is
\begin{equation}
\argmin_{\vM} \calL(\vD, \vA, \vM) .
\end{equation}
Note that the best-response mitigation does not actually depend on the detector configuration $\vD$.
\begin{revision}
To prove this, observe that the only term of $\calL(\vD, \vA, \vM)$ that depends on mitigation action $\vM$ is $\left(T_M(\vA, \vM) - T\right) \cdot \Delta_M$. Since this term does not depend on detector configuration $\vD$, neither does the best-response mitigation.
Intuitively, the explanation for this is that once the defender has detected the attack, it does not matter how it was detected (and all costs associated with detection are sunk).
\end{revision}

Given detector configuration $\vD$, a best-response attack is 
\begin{align}
&\argmax_{\vA} \calG(\vD, \vA, \vM) \,\big|_{\vM \in\, \argmin_{\vM'} \calL(\vD, \vA, \vM') } \nonumber \\
&=\argmax_{\vA} \calG(\vD, \vA, \vM) \,\big|_{\vM \in\, \argmin_{\vM'} \calG(\vD, \vA, \vM') + \sum_{i \in \calI_D} D_i \cdot C } \\
&=\argmax_{\vA} \calG(\vD, \vA, \vM) \,\big|_{\vM \in\, \argmin_{\vM'} \calG(\vD, \vA, \vM') } \\
&=\argmax_{\vA} \min_{\vM} \calG(\vD, \vA, \vM) .
\end{align}
Note that the attacker must anticipate the defender's mitigation action in the next stage; however, since the game is strategically equivalent to a zero-sum game, the attacker's problem simplifies to a maximin optimization.
\begin{revision}
Our threat model assumes a worst-case attacker, whose goal is to minimize the defender's utility.
This is a safe assumption since the defender's utility can only be higher if the attacker behaves differently.
In contrast, if our threat model assumed a particular attacker behavior, then the defender's strategy would be vulnerable to deviations from that assumed behavior. 
\end{revision}

Finally, an optimal (i.e., equilibrium) detector configuration is
\begin{align}
&\argmin_{\vD} \min_{\vM}  \calL(\vD, \vA, \vM) \,\big|_{\vA \in \argmax_{\vA'} \min_{\vM'} \calG(\vD, \vA', \vM') } \nonumber \\
&=\argmin_{\vD} \max_{\vA} \min_{\vM}  \calL(\vD, \vA, \vM) .
\end{align}
Note that the defender needs to anticipate the attacker's attack action in the next stage; however, since the game is strategically equivalent to a zero-sum game, the defender's problem simplifies to a minimax optimization (i.e., minimaximin if we consider the mitigation choice as well).

\section{Analysis}
\label{sec:analysis}

Even though we can express the players' optimal strategies as relatively simple maximin and minimax optimization problems, actually finding optimal strategies is computationally challenging due to the sizes of the players' strategy spaces.
Hence, we focus our analysis on the computational aspects of solving the transportation security game.
First, in Section~\ref{sec:complexity}, we show that finding an optimal action for the attacker is a computationally hard problem.
Although we study the complexity of only the attacker's problem in this paper, our computational-complexity argument could be easily extended to the defender's problem.
Due to the complexity of solving the game, we focus the remainder of this section on providing efficient heuristic algorithms:
we introduce a greedy heuristic for the attacker in Section~\ref{sec:greedy_attack},
and a metaheuristic search algorithm for the defender in Section~\ref{sec:search_configuration}.

\subsection{Computational Complexity}
\label{sec:complexity}

We begin our analysis by showing that the attacker's problem (i.e., finding a worst-case attack) is computationally hard.
Given a detector configuration $\vD$, the attacker's problem is to find an optimal attack~$\vA^*$ that maximizes $\min_{\vM} \calG(\vD, \vA^*, \vM)$.
Following the backward-induction approach, we assume that we have an oracle that finds the optimal mitigation action $\vM$ 
for any attack  $\vA$, and we study the attacker's problem by building on this oracle. 
In practice, the oracle can be replaced with, for example, a linear program for finding optimal traffic control.
Further, for ease of presentation, we overload the notation $\calG$ as follows
\begin{equation}
\calG(\vD, \vA) = \min_{\vM} \calG(\vD, \vA, \vM) .
\end{equation}

First, we formulate a decision version of the attacker's problem as follows.

\begin{definition}\emph{Attacker's Decision Problem}:
Given a transportation network, a budget $B$, a detector configuration  $\vD$, and a threshold gain $\calG^*$,
determine if there exists an attack~$\vA^*$ satisfying the budget constraint such that $\calG(\vD, \vA^*) > \calG^*$.
\end{definition}

We show that the above problem is computationally hard by reducing a well-known NP-hard problem, the Set Cover Problem, to the above problem.

\begin{definition}\emph{Set Cover Problem}:
Given a base set~$U$, a collection $\calC$ of subsets of $U$, and a number~$k$,
determine if there exists a subcollection $\calC' \subseteq \calC$ of at most $k$ subsets such that every element of~$U$ is contained by at least one subset in $\calC'$.
\end{definition}

The following theorem establishes the computational complexity of the attacker's problem.

\begin{theorem}
\label{thm:attackNPhard}
The Attacker's Decision Problem is $N\!P$-hard.
\end{theorem}

\begin{proof}
\begin{figure}
\centering
\begin{tikzpicture}[y=0.8,
  Cell/.style={draw, fill=white, rounded corners=0.25em, inner sep=0.5em, solid},
  Connection/.style={-{Stealth[scale=1.2]}},
]
\node (source) [Cell] {$r$};

\matrix (subsets) [right=of source, draw, dashed, row sep=12] {
\node (subsetsTop) [Cell] {}; \\
\node (subsetsSecond) [Cell] {}; \\
\node (subsetsDots) {$\vdots$}; \\
\node (subsetsBottom) [Cell] {}; \\
};
\node [above=0.2 of subsets] {$\calC$};
\node [below=0.2 of subsets] {\begin{tabular}{c}$\forall C \in \calC:$\\$Q_C = 1$\end{tabular}};

\matrix (elements) [right=2 of subsets, draw, dashed, row sep=12] {
\node (elementsTop) [Cell] {}; \\
\node (elementsSecond) [Cell] {}; \\
\node (elementsDots) {$\vdots$}; \\
\node (elementsBottom) [Cell] {}; \\
};
\node [above=0.2 of elements] {$U$};
\node [below=0.2 of elements] {\begin{tabular}{c}$\forall u \in U:$\\$Q_u = k + 1$\end{tabular}};

\node (sink) [Cell, right=of elements] {$s$};

\draw [Connection] (source) -- (subsetsTop);
\draw [Connection] (source) -- (subsetsBottom);
\draw [Connection] (subsetsTop) -- (elementsTop);
\draw [Connection] (subsetsSecond) -- (elementsSecond);
\draw [Connection] (subsetsSecond) -- (elementsBottom);
\draw [Connection] (subsetsBottom) -- (elementsBottom);
\draw [Connection] (elementsTop) -- (sink);
\draw [Connection] (elementsBottom) -- (sink);
\end{tikzpicture}
\caption{Illustration for the proof of Theorem~\ref{thm:attackNPhard}.}
\label{fig:attackNPhard}
\end{figure}
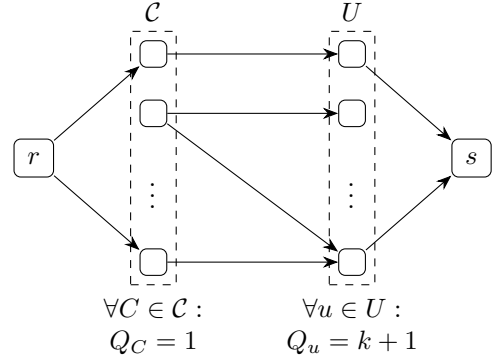

Given an instance of the Set Cover Problem (i.e., a set $U$, a collection $\calC$ of subsets, and a number~$k$), we construct an instance of the Attacker's Decision Problem as follows:
\begin{itemize}
\item let the transportation network be the following (see Figure~\ref{fig:attackNPhard} for an illustration):
\begin{itemize}
\item there is one source cell $r$, with $Q_r = k + 1$, $d_r^1 = k + 1$, and $d_r^t = 0$ for $t > 1$;
\item there is one sink cell $s$;
\item for every element $u \in U$, there is a merging cell~$u$;
\item for every subset $C \in \calC$, there is a diverging cell~$C$;
\item each diverging cell $C$ is connected to every merging cell $u \in C$;
\item for every cell $i$, $N_i = k + 1$ and $\delta_i = 1$;
\item for every merging cell $u$, $Q_u = k + 1$;
\item for every diverging cell $C$, $Q_C = 1$;
\end{itemize}
\item let the attacker's budget be $B = |U|$;
\item let the detector configuration be such that $\forall \vA: \Delta_D(\vD, \vA) \equiv 1$
\item let the default congestion be $T = 0$, let the congestion after the attack $T_A(\vA)$ be equal to the total travel time of the vehicles, and let the mitigation time be $\Delta_M = 0$;
\item let the threshold gain be $\calG^* = 3 (k + 1)$.
\end{itemize}
Clearly, the above reduction can be carried out in time that is polynomial in the size of the Set Cover Problem instance.

It remains to show that the above instance of the Attacker's Decision Problem has a solution $\vA^*$ if and only if the given instance of the Set Cover Problem has a solution $\calC'$.
Before we proceed to prove this equivalence, notice that the values $Q_r$, $N_i$ and $\delta_i$ for every cell $i$, and $Q_u$ for every merging cell $u$ will not play any role, since they are high enough to allow any traffic to pass through.
Furthermore, since $B = |U|$ and $\Delta_D \equiv 1$, the attacker will be able to reconfigure every traffic signal without decreasing detection time.
Hence, the attacker's problem is simply to pick the values $\hat{p}_{Cu}$ for every $u \in C$ so that the total travel time is at least $\calG^* = 3 (k + 1)$.

First, suppose that there exists a set cover $\calC'$ of size at most $k$.
Then, we construct an attack as follows: for every merging cell $u$, choose one diverging cell $C$ from $\calC'$ that is connected to $u$ (if there are multiple, then choose an arbitrary one), and let $\hat{p}_{Cu} = 1$.
We have to show that the total travel time in the transportation network is greater than $3 (k + 1)$ after the attack.
Since the distance between the source cell and the sink cell is 3 hops and there are $k + 1$ vehicles, all the vehicles must move one step closer to the sink in every time interval in order for the total travel time to be at most $3 (k + 1)$.
However, from the source cell, the vehicles may only move to the cells in~$\calC'$; otherwise, they would get ``stuck'' at one of the diverging cells that are not in $\calC'$.
Consequently, in the second time interval, at most $k$ of the $k + 1$ vehicles may move on, which means that the total travel time has to be greater than $3 (k + 1)$.

Second, suppose that there does not exist a set cover $\calC'$ of size at most $k$.
Then, we have to prove that there cannot exist an attack which increases the total travel time to more than $3 (k + 1)$. 
Firstly, we show that there exists an optimal attack which assigns either $0$ or $1$ to every $\hat{p}_{Cu}$.
To prove this, consider an attack in which there is a merging cell $v$ with a $\hat{p}_{Cv}$ value other than $0$ or $1$.
If none of its predecessor cells~$C$ has a positive $\hat{p}_{Cw}$ value for some other merging cell $w$, then the assignment for $v$ can clearly be changed to $0$ and $1$ values without changing the total  travel time.
Next, suppose that one (or more) of the predecessor cells $C$ of the merging cell has a positive $\hat{p}_{Cw}$ value for some other merging cell $w$. 
Then, the total travel time maximizing assignment is clearly one which assigns $\hat{p}_{Cv} = 1$ to a predecessor cell $C$ for which $\sum_{u \in C} \hat{p}_{Cu}$ is maximal, since this ``wastes'' the most ``merging capacity.''
Thus, for the remainder of the proof, it suffices to consider only attacks where every $\hat{p}_{Cu}$ value is either $0$ or $1$.

Now, consider an optimal attack $\vA^*$ against the transportation network, and let $\calC^*$ be the set of diverging cells $C$ for which there exists a merging cell~$u$ such that $\hat{p}_{Cu} = 1$.
Clearly, $\calC^*$ forms a set cover of~$U$ since for every element $u$, there is a subset $C \in \calC^*$ such that $u \in C$ (i.e., $C$ is connected to $u$).
From our initial supposition, it follows readily that the cardinality of set $\calC^*$ must be at least $k + 1$.
However, this also implies that the total travel time after the attack is equal to $3 (k + 1)$: in the second time interval, all $k + 1$ vehicles may move forward to the diverging cells in set $\calC^*$; in the third time interval, all the vehicles may again move forward to the merging cells (since every cell in $\calC$ has at least one ``enabled'' connection); and all the vehicles may leave the network by the next interval through the sink cell.
Since the total travel time after an optimal attack $\vA^*$ is equal to $T^* = 3 (k + 1)$, the attacker's problem does not have a solution.
Therefore, the constructed instance of the Attacker's Decision Problem has a solution if and only if the given instance of the Set Cover Problem has one, which concludes our proof.
\end{proof}

\subsection{Algorithms}
\label{sec:algorithms}

Mitigation---in our model---means adapting the schedule of uncompromised traffic signals given the schedule of compromised signals, which is equivalent to optimizing traffic control in a non-adversarial setting, with the compromised signals acting as fixed-schedule signals.
Since optimizing traffic control in non-adversarial settings has been studied in prior work, we focus on providing efficient algorithms for solving the first two stages of the game.

\subsubsection{Greedy Algorithm for Attacks}
\label{sec:greedy_attack}

Since the attacker's problem is $N\!P$-hard, we cannot hope for a polynomial-time algorithm that always finds a worst-case attack (unless $P = N\!P$).
Hence, to provide an alternative to computationally infeasible exhaustive search, we turn our attention to designing an efficient heuristic algorithm.

\begin{algorithm}[h!]
\caption{Polynomial-Time Greedy Heuristic for Finding an Attack}
\label{alg:heuristic}
\KwData{transportation network security game, detector configuration $\vD$}
\KwResult{attack $\vA^*$}
$\calI_A^* \gets \emptyset$,
$\hat{\vp}^* \gets \vp$ \;
\For{$b = 1, \ldots, B$}{
  $\calI_A' \gets \calI_A^*$,
  $\hat{\vp}' \gets \hat{\vp}^*$ \;
  \For{$i \in \calI$}{
    $\calI_A \gets \calI_A^* \cup \{i\}$ \;
    \For{$k \in \Gamma^{-1}(i)$}{      
      $\forall l, j: ~ \hat{p}_{lj} \gets \begin{cases}
        1 & \text{ if } j = i \wedge l = k \\
        0 & \text{ if } j = i \wedge l \in \Gamma^{-1}(i) \setminus \{k\} \\
        \hat{p}^*_{lj} & \text{ otherwise.}
      \end{cases}$ \;
      \If{$\calG(\vD, (\calI_A, \hat{\vp})) \geq \calG(\vD, (\calI_A', \hat{\vp}'))$}{
        $\calI_A' \gets \calI_A$,
        $\hat{\vp}' \gets \hat{\vp}$ \;
      }
    }
  }
  $\calI_A^* \gets \calI_A'$,
  $\hat{\vp}^* \gets \hat{\vp}'$ \;
}
\textbf{output} $\vA^* = (\calI_A^*, \hat{\vp}^*)$
\end{algorithm}

The attacker's problem can be viewed as the composition of two problems: finding a subset $\calI_A$ of at most $B$ signalized intersections and finding new inflow proportions $\hat{p}_{ki}$ for the cells $i \in \calI_A$.
For finding a subset $\calI_A$, we propose to use a greedy heuristic, which starts with an empty set and adds new cells to it one-by-one, always picking the one that leads to the greatest increase in the attacker's gain.
Finding new inflow proportions $\hat{p}_{ki}$ is particularly challenging, since the set of possible choices is continuous.
However, we observe that in most networks, the worst-case configuration is an ``extreme'' one, which assigns proportion $\hat{p}_{ki} = 1$ to one predecessor cell $k$ and proportion $\hat{p}_{ji} = 0$ to every other predecessor cell~$j$~\cite{laszka2016vulnerability}. 
\begin{revision}
In fact, we have tested this property on hundreds of networks that we generated according to the Grid model with Random Edges (see Section~\ref{sec:numerical_game}), resembling road networks from the U.S. and Europe, and we have not found a single counterexample.
Based on these numerical results, we conjecture that this property holds in general, that is, for any network.\footnote{\revision{We leave the theoretical proof of this claim for future work.}}\end{revision} %
Hence, for every new cell $i$ added to the set of attacked intersections, we propose to search over the possible extreme configurations by iterating over the predecessors of cell $i$.
Based on the above ideas, we formulate Algorithm~\ref{alg:heuristic}. 

It is fairly easy to see that we can implement Algorithm~\ref{alg:heuristic} as a polynomial-time algorithm.
Due to the three nested iterations, the running time of Algorithm~\ref{alg:heuristic} is $O\big(B \cdot |\calI|$ $\cdot \left(\max_{i \in \calI} |\Gamma^{-1}(i)|\right)\big)$ times the running time of computing~$\calG$.
Since 
we can compute $\calG(\vD, \vA)$ for any attack~$\vA$ using a linear program, it follows readily that the running time of the algorithm can be upper bounded by a polynomial function of the input size (i.e., size of the transportation network and budget~$B$).
We can formally state this observation as the following proposition.

\begin{proposition}
With a polynomial-time oracle for computing $\calG$,
the running time of Algorithm~\ref{alg:heuristic} is a polynomial function of the input size.
\end{proposition}

\subsubsection{Metaheuristic Search Algorithm for Detector Configuration}
\label{sec:search_configuration}

Next, we present an algorithm for finding a detector configuration (i.e., false-positive rates) based on a metaheuristic approach.
In particular, we use \emph{simulated annealing} to find a near-optimal detector configuration $\vD$. 
The basic idea  of this approach is to start with an arbitrary configuration $\vD$, which we then improve iteratively. 
In each iteration, we generate a new solution $\vD'$ in the neighborhood of $\vD$. 
If the new configuration $\vD'$ is better in terms of minimizing the defender's loss (against an attacker playing a best response), then the current configuration $\vD$ is replaced with the new one. 
On the other hand, if the new configuration $\vD'$ increases the defender's loss, then new configuration replaces the current one with only a small probability. 
This probability depends on the difference between the two solutions in terms of loss as well as a parameter commonly referred to as ``temperature,'' which is a decreasing function of the number of iterations. 
These random replacements prevent the search from ``getting stuck'' in a local minimum. 
The algorithm is presented below as Algorithm~\ref{alg:sim_ann}. 

\begin{algorithm}[h!]
\caption{Polynomial-Time Metaheuristic for Finding a Detector Configuration}
\label{alg:sim_ann}
\KwData{transportation network security game, iterations $k_{\max}$, initial temperature $T_0$, cooling parameter $\beta$} 
\KwResult{detector configuration $\vD^*$}
$\vD \gets \vect{1}$ \;
$L \gets \max_{\vA} \calG(\vD, \vA) + \sum_{i \in \calI_D} D_i \cdot C$ \;
\For{$k = 1, \ldots, k_{\max}$}{
  $\vD' \gets \mathtt{Perturb}(\vD)$ \;
  $L' \gets \max_{\vA} \calG(\vD, \vA) + \sum_{i \in \calI_D} D'_i \cdot C$ \;
  $T \gets T_0 \cdot e^{-\beta k}$ \;
  $pr \gets e^{(L' - L) / T}$ \;
  \If{$(L' < L) \; \vee \; (\mathtt{rand}(0, 1) \le pr)$}{
    $\vD \gets \vD'$,  $L \gets L'$ 
  }
}
\textbf{output} $\vD$
\end{algorithm}

In Algorithm \ref{alg:sim_ann}, $\mathtt{Perturb}(\vD)$ picks a random configuration $\vD'$ from the neighborhood of $\vD$. 
In particular, we implement $\mathtt{Perturb}(\vD)$ as choosing a value for each $D_i'$ uniformly at random from $[D_i \cdot (1 - \varepsilon), D_i \cdot (1 + \varepsilon)]$, where~$\varepsilon$ is a small constant (e.g., $0.1$).
For solving $\max_{\vA} \calG(\vD, \vA)$ in practice, we can use the greedy heuristic (Algorithm~\ref{alg:heuristic}).
The temperature $T$ is decreasing exponentially with iteration number $k$, and the rate of the decrease is controlled by the ``cooling'' parameter $\beta$.
Finally, we note that a simpler algorithm could also be obtained, in which $\vD$ is updated with $\vD'$ in each iteration if and only if $\vD'$ is strictly better than $\vD$. 
This heuristic search, commonly known as \textit{hill climbing}, also works well for our problem; however, Algorithm \ref{alg:sim_ann} gives better results.

\section{Anomaly-based Detector}
\label{sec:detector}

Now, we introduce a traffic-anomaly based detector against stealthy attacks that tamper with traffic control.
The core idea of anomaly-based detection is to build a probabilistic model of normal traffic conditions, which can then be used to estimate the likelihood that observed traffic conditions are normal.
\begin{revision}
Note that we must employ a probabilistic model to account for the uncertainty in parameter values since many parameters (e.g., traffic demand) can only be estimated in practice.
\end{revision}
We can estimate the likelihood that the observed traffic is normal as the probability that our model of normal traffic would generate the observed traffic.
We can then compare the likelihood value to a threshold, and if the likelihood is lower, we can raise an alarm.
In our detector, the model of normal traffic is based on Gaussian processes, which have been successfully used in prior work for traffic volume forecasting~\cite{xie2010gaussian,chen2012decentralized}.
\begin{revision}
Note that we cannot use macroscopic traffic models, such as the cell transmission model, to detect attacks because these models abstract  away details for the sake of tractability (e.g., inflow proportions instead of actual traffic light schedules); however, such details can be crucial for the detection of stealthier attacks that alter traffic control only slightly.
\end{revision}

\subsection{Gaussian Processes}

We begin giving a very brief overview of Gaussian processes. 
For a comprehensive discussion of Gaussian processes in machine learning, we refer the reader to~\cite{rasmussen2006gaussian}.

In principle, Gaussian processes are an extension of multivariate Gaussian distributions to infinite collections of random variables.
Formally, a Gaussian process is a stochastic process such that any finite collection of variables $(X_1, \ldots, X_n)$ follows a multivariate Gaussian distribution.
A Gaussian process is typically described using a \emph{mean function} 
\begin{equation}
\mu(X) = \mathbb{E}(X)
\end{equation}
and a \emph{covariance function} 
\begin{equation}
k(X_1, X_2) = \mathbb{E}\left[(X_1 - m(X_1)) (X_2 - m(X_2))\right] .
\end{equation}
The covariance function is often chosen to be some well-known kernel function, such as squared exponential, whose parameters can be estimated from a training dataset $(x_1,$ $\ldots, x_n)$.
A common application of Gaussian processes is regression: given the values of a set of training variables $(x_1, \ldots, x_n)$, we can easily compute the expected value and variance of a target variable $Y$ using the mean and covariance functions.
Gaussian-process based regression models have been successfully applied to a wide range of problem, such as 
traffic volume forecasting~\cite{xie2010gaussian,chen2012decentralized},
spatial modeling of extreme snow depth~\cite{blanchet2011spatial},
wind power forecasting~\cite{kou2013sparse},
estimation of water chlorophyll concentration~\cite{bazi2012improved}, and
spectrum sensing~\cite{nevat2012location}.




\subsection{Model}


We assume that traffic sensors, such as induction loop sensors, have been deployed for monitoring the transportation network.
Since modeling an entire network would be computationally challenging---and would certainly not scale well---we divide sensors into subsets, and we build a separate model and detector for each one of these subsets.
For example, traffic sensors that are deployed next to the same intersection $i \in \calI$ may be grouped together and provide traffic data for one detector (see, e.g.,  Figure~\ref{fig:intersection} in Section~\ref{sec:case-detector}).
The outputs of all the detectors deployed in a transportation network can then be combined together to form a single detector for the entire network.

We assume that sensors measure and report traffic values, such as traffic flow or occupancy, in fixed-length intervals (e.g., report one measurement value for every 15-minute intervals).\footnote{Note that the length of these measurement intervals is independent of the time intervals of the cell-transmission model. However, for ease of presentation, we will reuse notation $t$ to identify measurement intervals.}
In our Gaussian-process model, we model each one of these measurements as a random variable. 
Formally, for every sensor $s$ and time interval $t$, there exists a random variable $X_s^t$ whose value is equal to the traffic value measured by sensor $s$ in interval $t$.
Hence, our model has a discrete but---due to the time dimension---potentially infinite set of variables.

A key part of modeling is establishing mean and covariance functions for these variables.
For each sensor $s$, we model the mean values $\mu(X_s^t)$ of variables $X_s^t$ using a periodic function: 
\begin{equation}
\mu\left(X_s^t\right) \equiv \mu\left(X_s^{t + P}\right) ,
\end{equation}
where $P$ is the length of the period.
We call this time period, which is measured in number of discrete time intervals, the \emph{detector window} $W$.
Similarly, for sensors $s_1$ and $s_2$, we model the covariance values $k\left(X_{s_1}^{t_1}, X_{s_2}^{t_2}\right)$ between variables $X_{s_1}^{t_1}$ and $X_{s_2}^{t_2}$ using a periodic function: 
\begin{equation}
k\left(X_{s_1}^{t_1}, X_{s_2}^{t_2}\right) \equiv k\left(X_{s_1}^{t_1 + P}, X_{s_2}^{t_2 + P}\right) .
\end{equation}
Further, we assume that $k\left(X_{s_1}^{t_1}, X_{s_2}^{t_2}\right) \equiv 0$ if $|t_1 - t_2| > P$.
To train the model, the actual values of these functions must be learned from traffic values observed during normal operation.

\subsection{Training}

Before training, 
our model has $S \cdot P + \frac{S \cdot (S - 1) \cdot P \cdot (2P + 1)}{2} + S \cdot P \cdot (2P + 1)$ unknown values, where $S$ is the number of sensors:
\begin{itemize}
\item $S \cdot P$ mean values: for each sensor $s$, the  mean function $\mu$ can be described by $P$ values since its period length is $P$; 
\item $\frac{S \cdot (S - 1) \cdot P \cdot (2P + 1)}{2}$ covariance values: for each distinct pair of sensors $s_1$ and $s_2$, the covariance function $k$ can be described by $P \cdot (2 \cdot P + 1)$ values since its period length is $P$, the maximum difference between $t_1$ and $t_2$ is $P$, and covariance values are symmetric; 
\item and  
$S \cdot P \cdot (2P + 1)$ variance values: for each sensor $s$, the covariance function $k$ can be described by $P \cdot (2 \cdot P + 1)$.
\end{itemize}
Training the model means learning these mean and covariance values for normal traffic. 
In practice, we can train the model by observing sensor measurements $(x_{s_1}^{t_1}, x_{s_1}^{t_2}, x_{s_2}^{t_1}, \ldots)$ of traffic under normal conditions, and then simply estimating the most likely mean and covariance values from these observations (i.e., maximum likelihood estimation).

\subsection{Detection}

Once we have trained the model, we can use it to detect attacks against traffic control.
First, we take sensor measurements $(x_s^t, \ldots)$, which are observed in the network that might be under attack, and we use the Gaussian process to compute the likelihood of these measurement values being generated by our model of normal traffic. 
\footnote{Due computational limitations, we restrict detection to observations from a single detector window (i.e., measurement values from some range $(t, t + P - 1)$). 
If more observations are available, we can evaluate the detector multiple times.}
Since the measurement values are continuous, we can use the probability density of the Gaussian distribution $(X_s^t, \ldots)$ at $(x_s^t, \ldots)$ as the likelihood value.
We then interpret this likelihood as the likelihood of the network operating under normal control (i.e., not being under attack). 
Finally, we compare the likelihood to a \emph{detector threshold} $\tau_i$, and raise an alarm if the likelihood is lower than the threshold.
Thus, our detector has two parameters, the detector window $W$ and the detector threshold $\tau_i$, which together determine the rate of false alarms $D_i$ and the detection delay.
\begin{revision}
Note that the detector threshold~$\tau_i$ and the rate of false alarms $D_i$ are closely related to each other: lower thresholds result in fewer false alarms since more observations are accepted as likely; and vice versa.
Hence, we can express $\tau_i(D_i)$ and $D_i(\tau_i)$ as increasing functions, and we may specify the configuration of the detector either as the desired false-positive rate $D_i$ or as the threshold $\tau_i$.
We chose to use representation $D_i$ in our game-theoretic model and analysis for ease of presentation. 
The exact relation $\tau_i(D_i)$ can be determined experimentally by evaluating the detector with various configurations on normal~traffic.
\end{revision}

\section{Numerical Results}
\label{sec:numerical}

In this section, we present numerical results on our heuristic algorithms from Section~\ref{sec:analysis} and our anomaly-based detector from Section~\ref{sec:detector}.

\subsection{Anomaly-based Attack Detection}
\label{sec:case-detector}

We begin by training and evaluating our detector based on simulated flows of traffic under normal conditions and under attacks.
We will then use the results of this evaluation to instantiate our game-theoretic model.
In particular, we will use the measured false-alarms rate and detection delay values as numerical parameters for our game-theoretic model.

\subsubsection{Setup}

\begin{figure}[t!]
\centering
\includegraphics[width=\columnwidth]{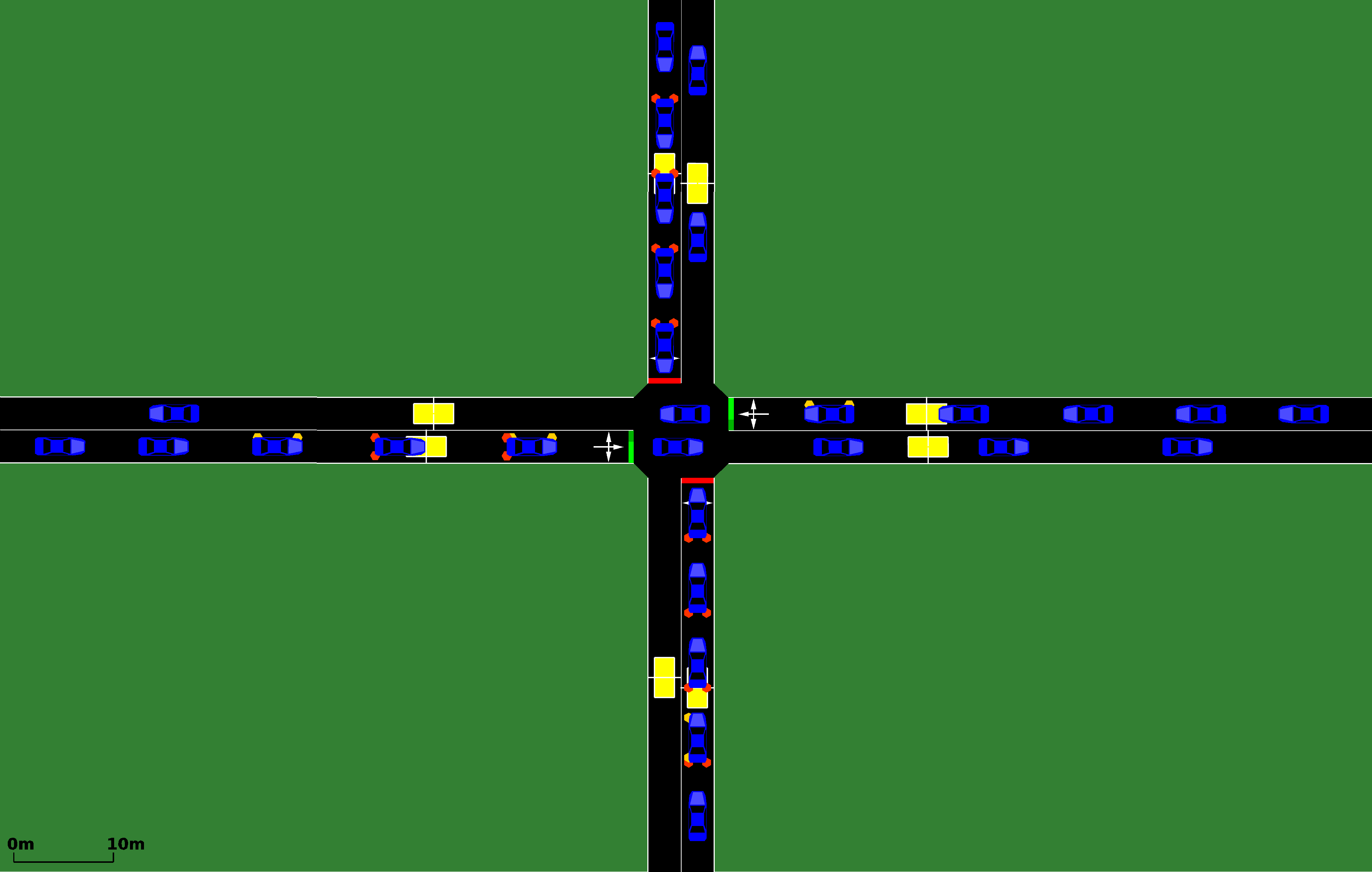}
\caption{Intersection used for evaluating the detector. Yellow rectangles represent induction-loop sensors.}
\label{fig:intersection}
\end{figure}

For these experiments, we simulate the signalized four-way intersection shown by Figure~\ref{fig:intersection} using SUMO (Simulation of Urban MObility)~\footnote{\url{http://sumo.dlr.de/wiki/Main_Page}}, a well-known and widely-used micro simulator~\cite{krajzewicz2002sumo,behrisch2011sumo}.
Signals are deployed on both the incoming and outgoing lanes (represented by yellow rectangles in the figure), and these signals measure traffic flow in 15-second intervals.
We generate one month of traffic data with the original signal schedule (45 seconds for both roads, including left-turning and yellow phases) for training the Gaussian-process based model.
In these simulations, vehicles enter the intersection from all four directions, and each car turns left, continues straight, or turns right with probability 5.3\%, 73.7\%, and 21.1\%, respectively.
From each direction, 0.19 vehicles arrive each second on average. 
We also generate one month of test traffic-data with the original schedule for measuring the false-positive rate of the detector.
Finally, for each of the attacks considered below, we generate one day of traffic data, which includes 1 hour with the original schedule and then 23 hours with the tampered schedule.

\begin{revision}
\paragraph{Posterior Predictive Check}
To confirm that our Gaussian-process based traffic model fits observations well, we perform \emph{posterior predictive checking}~\cite{gelman2013bayesian}.
The idea of posterior predictive checking is to draw simulated samples from the joint posterior predictive distribution and compare them to the observed sample.
If there were significant systematic differences between the simulated and observed samples, that would indicate that our model did not fit well.
Note that the applicability of our model is also demonstrated by the low false-positive rate exhibited by our detector.

\begin{table*}[h]
\centering
\begin{revision}
\caption{Posterior Predictive Checking $p$-Values}
\label{tab:postpred}
\begin{tabular}{|l|rrrrr|}
\hline
Sensor & \multicolumn{5}{c|}{Test statistic $T$} \\
 & mean & variance & median & quantile $Q(0.3)$ & quantile $Q(0.7)$ \\
\hline
east in & 0.54 & 0.52 & 0.54 & 0.79 & 0.49 \\
east out & 0.51 & 0.47 & 0.66 & 0.78 & 0.36 \\
south in & 0.46 & 0.53 & 0.49 & 0.79 & 0.26 \\
south out & 0.49 & 0.48 & 0.66 & 0.79 & 0.28 \\
west in & 0.54 & 0.52 & 0.55 & 0.80 & 0.49 \\
west out & 0.50 & 0.47 & 0.66 & 0.78 & 0.36 \\
north in & 0.47 & 0.52 & 0.64 & 0.84 & 0.57 \\
north out & 0.49 & 0.48 & 0.66 & 0.79 & 0.28 \\
\hline
\end{tabular}
\end{revision}
\end{table*}

The significance of difference can be quantified as the classical $p$-value~\cite[Chapter~6.3]{gelman2013bayesian}:
\begin{equation}
    p = \Pr\left[ T\left(x^{\textnormal{rep}}\right) \geq T\left(x\right) \,\middle|\, \mu, k\right] ,
\end{equation}
where $x^{\textnormal{rep}}$ is the replicated data generated according to our model, $x$ is the observed data, $\mu$ and $k$ are the model parameters (see Section~\ref{sec:detector}), and $T$ is a test statistic.
In our checks, we compute $p$ values for various standard test statistics, including mean, variance (with Bessel's correction), median, and quantiles.
Note that a perfectly fitting model will yield $p$-values around $0.5$.
To reliably estimate the probability $p$, we generate and evaluate the test statistics on 10,000 replicated samples (drawn independently according to our trained model).

Table~\ref{tab:postpred} shows the $p$-values for various statistical tests ($p$-values around $0.5$ indicate a perfect fit).
Since our samples are multidimensional (i.e., one value for each sensor around the intersection), we apply the statistical tests to the marginal distributions corresponding to the individual sensors.
We list sensors in clockwise order, starting with the sensors on the incoming and outgoing lanes of the road eastward of the intersection (see Figure~\ref{fig:intersection}), denoted `east in' and `east out.'
For mean and variance statistics, we see that our model produces an almost perfect fit for all sensors, which is important since these play key role in determining likelihood, on which our detector is built.
Further, we see that our model produces a good fit for the median statistic as well, which indicates that there is no significant asymmetry that could not be captured by our model.
Finally, we test quantiles $Q(0.3)$ and $Q(0.7)$, and see a reasonable fit for most sensors.
The worst fit is for sensor `north in' (i.e., inward lane of the road northwards) and $Q(0.3)$; however, we see an almost perfect fit for the same sensor for $Q(0.7)$, indicating a skewed distribution.
Next, we study the detection performance of our model, showing that observations under normal conditions and under attacks exhibit significantly different likelihood values.
\end{revision}

\subsubsection{Detector Configuration}
We first consider the defender's problem of balancing the number of false-positive errors and the detection delay.
For this experiment, we assume an attack which changes 4.4\% of the traffic-signal schedule.

\begin{figure}[h!]
\centering
\begin{tikzpicture}
\begin{axis}[
  width=\columnwidth,
  height=0.8\columnwidth,
  font=\small,
  xmin=5,
  xmax=70,
  ymax=800,
  xlabel={Detection delay [minutes]},
  ylabel={Number of false positives [per month]},
  grid=both,
  xtick={10, 20, 30, 40, 50, 60},
]
\addplot[] table[x=delay_minutes, y=false_positives, comment chars={\%}, col sep=comma] {software/anomaly_detection/plots/pareto_detector_4.csv};
\end{axis}
\end{tikzpicture}
\caption{Trade-off between false-positive rate and detection delay.}
\label{fig:detector_tradeoff}
\end{figure}
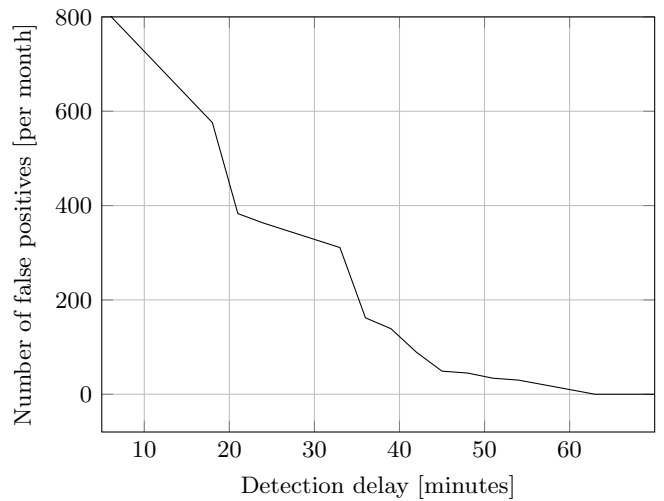

Figure~\ref{fig:detector_tradeoff} shows the trade-off between the false-positive rate and the detection delay.
Each point on the curve is a Pareto optimal point that is attainable with some detector window $W$ and threshold~$\tau$.
The figure shows that with a negligible false-positive rate, even the stealthy attack considered in this example can be detected in approximately one hour.
The configuration of the detector when the false-positive rate reaches zero is detector window being equal to $W = 48$ minutes and log-likelihood threshold being equal to $\ln \tau = -112.58$.

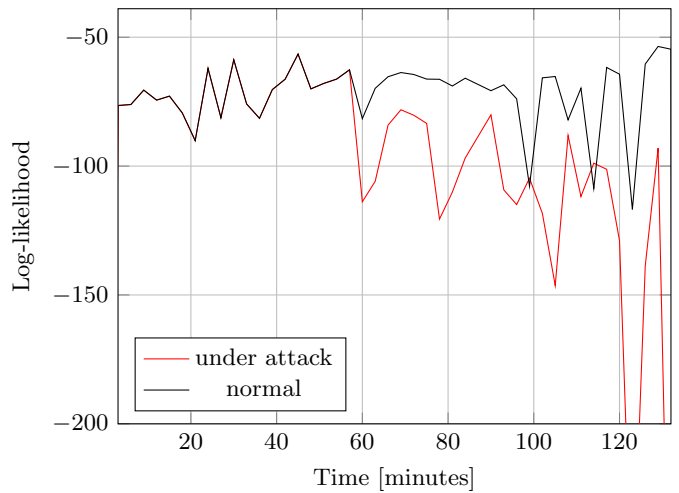
\begin{figure}[h!]
\centering
\begin{tikzpicture}
\begin{axis}[
  width=\columnwidth,
  height=0.8\columnwidth,
  font=\small,
  xmin=3,
  xmax=132,
  xlabel={Time [minutes]},
  ylabel={Log-likelihood},
  grid=both,
  legend pos=south west,
  ymin=-200,
]
\addplot[red] table[x expr=\thisrowno{0}*3, y=4, comment chars={\%}, col sep=comma] {software/anomaly_detection/results/likelihoods.csv};
\addlegendentry{under attack};
\addplot[] table[x expr=\thisrowno{0}*3, y=0, comment chars={\%}, col sep=comma] {software/anomaly_detection/results/likelihoods.csv};
\addlegendentry{normal};
\end{axis}
\end{tikzpicture}
\caption{Likelihood values output by the Guassian-process based model.}
\label{fig:detector_likelihood}
\end{figure}

Figure~\ref{fig:detector_likelihood} shows the likelihood values output by the Gaussian-process model for traffic data resulting from the original and the attacked traffic-signal schedules.
For this figure, we set the detector window to be $W = 3$ minutes, which results in highly variable likelihood values.
The figure shows that after one hour (i.e., when the attack starts), the likelihood values for the tampered schedule become much lower than for the original one.
In other words, the detector correctly estimates that the traffic with tampered schedule is less likely to be~normal.

\subsubsection{\revision{Stealthy} Attacks}
Next, we consider the attacker's problem of balancing stealthiness and impact. \begin{revision}
If stealthy attacks could avoid detection for extended periods of time while having substantial impact on traffic, they could pose a significant threat to the transportation network. 
To show that stealthy attacks are not effective against our detector, we compare a wide range of attacks, from stealthy ones that change control only slightly to non-stealthy ones that change control fundamentally.
To maximize the advantage of stealthiness, we consider a low detection threshold, which allows attacks to remain undetected for long periods of time.
In particular,
\end{revision}
for this experiment, we assume a detector window of $W = 48$ minutes and a log-likelihood threshold of $\ln\tau = -112.58$, which result in zero false-positive rate for the one-month test interval (see the discussion of Figure~\ref{fig:detector_tradeoff}).

\begin{figure}[h!]
\centering
\begin{tikzpicture}
\begin{axis}[
  width=0.94\columnwidth,
  height=0.8\columnwidth,
  font=\small,
  axis y line*=left,
  xlabel={Attack magnitude [\% of signal schedule changed]},
  ylabel={Detection delay [minutes]},
  grid=both,
]
\addplot[] table[x=attack, y=delay, comment chars={\%}, col sep=comma] {software/anomaly_detection/plots/attacks.csv};
\end{axis}
\begin{axis}[
  width=0.94\columnwidth,
  height=0.8\columnwidth,
  font=\small,
  axis y line*=right,
  ylabel={Impact rate [\% of traffic blocked]},
  legend style={at={(0.5,0.97)},anchor=north},
    xtick={},
]
\addlegendimage{no markers};
\addlegendentry{delay};
\addplot[dashed] table[x=attack, y=impact, comment chars={\%}, col sep=comma] {software/anomaly_detection/plots/attacks.csv};
\addlegendentry{impact};
\end{axis}
\end{tikzpicture}
\caption{Impact and detection delay for various attacks.}
\label{fig:detector_attacks}
\end{figure}
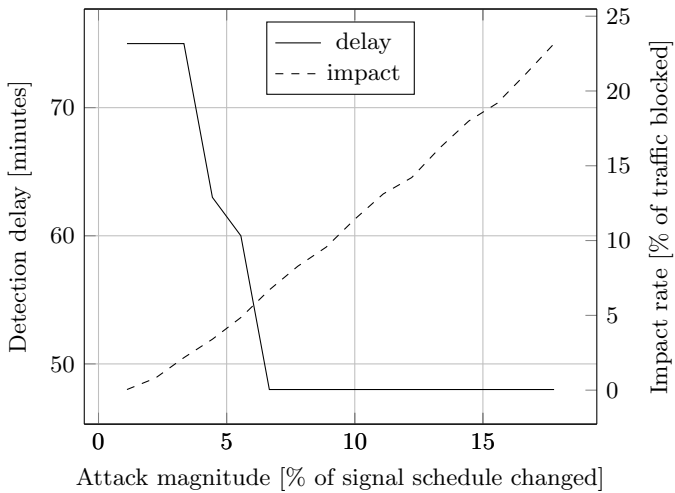

Figure~\ref{fig:detector_attacks} shows detection delay and impact for attacks of various magnitudes.
The impact of an attack is measured as the fraction of traffic that is ``blocked'' by the attack, i.e., the decrease in the number of vehicles passing through the intersection compared to the normal traffic-signal schedule;
the magnitude of an attack is the fraction of the traffic-signal schedule that is modified by the attacker.
The figure shows that attacks with higher magnitude may be less stealthy (i.e., detected earlier), but they cause much more significant impact.
In fact, the total impact of attacks, measured as the number of vehicles that could not pass through the intersection due to the attack until its detection, is a strictly increasing function of the attack magnitude.
This means that our detector can make stealthy attacks essentially pointless in this example.

\subsection{Multi-Stage Security Game Strategies}
\label{sec:numerical_game}

Next, we provide numerical results on the algorithms that we proposed in Section~\ref{sec:algorithms} for finding strategies in practice.
We first compare the proposed greedy heuristic for finding attacks (Algorithm~\ref{alg:heuristic}) to an exhaustive search, and then study the metaheuristic search algorithm for finding detector configurations (Algorithm~\ref{alg:sim_ann}).

\subsubsection{Setup}

To provide meaningful numerical results, we have to evaluate our algorithms on a large number of transportation networks.
Every point plotted in the figures of this subsection represents a mean value computed over a large number of random networks with the same parameters. 
To obtain these networks, we use the Grid model with Random Edges (GRE) to generate random network topologies~\cite{peng2012random}, which closely resemble real-world transportation networks.
For a detailed description of this model, we refer the reader to~\cite{peng2012random,peng2014random}.

We set both the width and height of the generated grids to be 4, and let the bottom-left corner be a source and the upper-right corner be a sink.
For the parameters controlling the randomness of the generation, we use the values from~\cite{peng2012random}, which were derived from measurements on actual road networks from the USA.
We let the inflow at the source cell be $d^0 = 8$, $d^1 = 12$, $d^2 = 8$, and $d^{t} = 0$ for $t \geq 3$.
For every other cell $i$, we let the parameters be $Q_i = 6$, $\delta_i = 1.0$, and $N_i = 10$. 
Finally, we let every merging cell be a signalized intersection, and optimize the inflow proportions for every intersection using a linear program.

We assume that there is an anomaly detector deployed in each intersection (i.e., $\calI_D = \calI$).
We also assume that every one of these detectors exhibits the false-positive rate and detection delay characteristics observed in Section~\ref{sec:case-detector}.
In other words, for each intersection, the defender chooses one of the Pareto optimal configurations that were identified in the experiments of Section~\ref{sec:case-detector} by choosing a false-positive rate $D_i$; the delay of this detector is then determined by the magnitude of the attack against the corresponding intersection.
Finally, we assume that the attack is detected as soon as one detector raises an alarm, that attack mitigation takes $\Delta_M = 20$ minutes, and that congestion levels $T$, $T_A$, and $T_M$ are measured in travel time.

\subsubsection{Attacks}

We begin by comparing the greedy attack heuristic to an exhaustive search.
To perform an exhaustive search, we quantize the space of possible schedules for each intersection, so that we have a finite and discrete search space.
For this experiment, we assume that the defender uses a detector configuration that sets the false-positive rate of every intersection $i \in \calI_D$ to $D_i = 1$.

\begin{figure}[h!]
\centering
\begin{tikzpicture}
\begin{axis}[
  width=\columnwidth,
  height=0.8\columnwidth,
  font=\small,
  xlabel={Budget $B$},
  ylabel={Impact $\calG$},
  xtick={1,2,3,4},
  grid=both,
  legend pos=south east,
]
\addplot[dashed] table[x=budget, y=exhaustive.impact, comment chars={\%}, col sep=comma] {software/multi_stage_strategies/plots/attack_algorithms.csv};
\addlegendentry{exhaustive search};
\addplot[] table[x=budget, y=greedy.impact, comment chars={\%}, col sep=comma] {software/multi_stage_strategies/plots/attack_algorithms.csv};
\addlegendentry{greedy (Algorithm~\ref{alg:heuristic})};
\end{axis}
\end{tikzpicture}
\caption{Impact resulting from attacks found by exhaustive and greedy searches.}
\label{fig:attack_impact}
\end{figure}
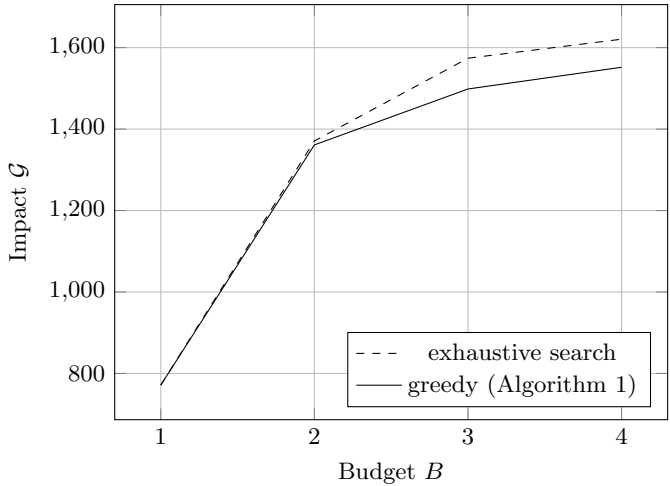

Figure~\ref{fig:attack_impact} shows the impact of attacks found by exhaustive and greedy (Algorithm~\ref{alg:heuristic}) searches for various budget values.
The vertical axis shows the total impact $\calG$ of the attacks, which includes impact that was caused both before and after detection. 
The figure shows that the attacks found by the greedy search are very close to the ones found by the exhaustive search in terms of total impact, 
with the largest difference being 5\%.

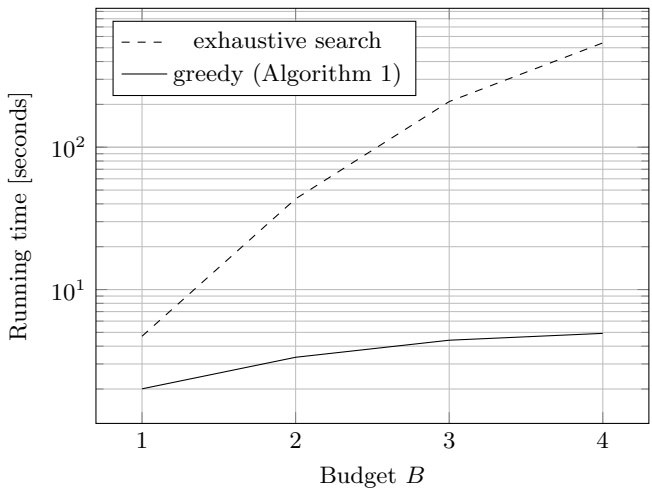
\begin{figure}[h!]
\centering
\begin{tikzpicture}
\begin{axis}[
  width=\columnwidth,
  height=0.8\columnwidth,
  font=\small,
  ymode=log,
  xlabel={Budget $B$},
  ylabel={Running time [seconds]},
  xtick={1,2,3,4},
  grid=both,
  legend pos=north west,
]
\addplot[dashed] table[x=budget, y=exhaustive.runtime, comment chars={\%}, col sep=comma] {software/multi_stage_strategies/plots/attack_algorithms.csv};
\addlegendentry{exhaustive search};
\addplot[] table[x=budget, y=greedy.runtime, comment chars={\%}, col sep=comma] {software/multi_stage_strategies/plots/attack_algorithms.csv};
\addlegendentry{greedy (Algorithm~\ref{alg:heuristic})};
\end{axis}
\end{tikzpicture}
\caption{Running time of exhaustive and greedy searches.}
\label{fig:attack_runtime}
\end{figure}

Figure~\ref{fig:attack_runtime} compares the greedy heuristic (Algorithm~\ref{alg:heuristic}) to the exhaustive search in terms of running time.
Note that we used fairly small problem instances for our experiments in order to be able to apply the algorithms to a large number of networks.
We observe that the running time of the greedy heuristic is much lower than that of the exhaustive search, and it grows slower as the attacker's budget increases.

\subsubsection{Detector Configuration}

Next, we evaluate the metaheuristic search algorithm for finding detector configurations.
We compare our strategic configurations to a non-strategic baseline represented by uniform configurations, which assign the same false-positive rate to all detectors.
We find quasi-optimal uniform configurations using the same search algorithm, but restricting the search space to a single scalar value, which is used for all detectors.
For these experiments, we let the attacker's budget be enough to compromise $B = 2$ intersections; we assume that the attacker always mounts a best-response attack; and we let the unit cost of false positives be equal to $C = 10$.

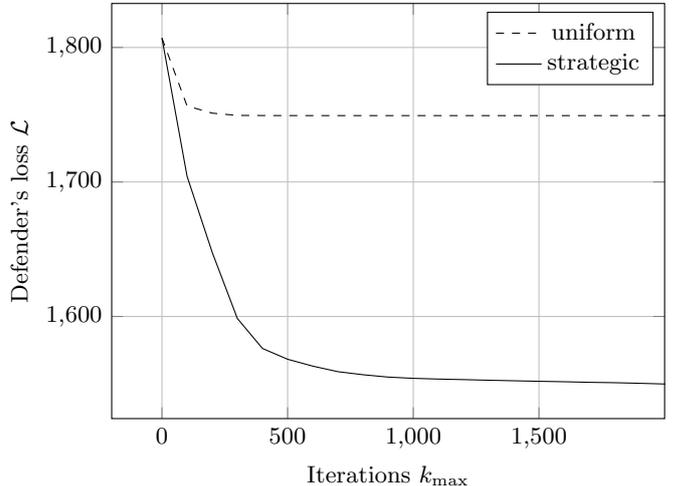
\begin{figure}[h!]
\centering
\begin{tikzpicture}
\begin{axis}[
  width=\columnwidth,
  height=0.8\columnwidth,
  font=\small,
  xmax=2000,
  xlabel={Iterations $k_{\max}$},
  ylabel={Defender's loss $\calL$},
  grid=both,
  xtick={0, 500, 1000, 1500},
]
\addplot[dashed] table[x=iterations, y=uniform.total, comment chars={\%}, col sep=comma] {software/multi_stage_strategies/plots/configuration_algorithms.csv};
\addlegendentry{uniform};
\addplot[] table[x=iterations, y=strategic.total, comment chars={\%}, col sep=comma] {software/multi_stage_strategies/plots/configuration_algorithms.csv};
\addlegendentry{strategic};
\end{axis}
\end{tikzpicture}
\caption{Defender's loss resulting from uniform and strategic detector configurations output by the metaheuristic search algorithm.}
\label{fig:detector_configuration}
\end{figure}

Figure~\ref{fig:detector_configuration} shows the defender's total loss---which includes both the cost of investigating false alarms and the total impact of the attack---with strategic and uniform detector configurations.
The horizontal axis shows the number of iterations $k_{\max}$ for which the search algorithms was run.
We can learn two important lessons from this figure.
First, strategic thresholds result in much lower losses than uniform ones, which suggests that game-theoretic optimization can have a significant practical impact.
Second, losses decrease rapidly in the first 100 or 500 hundred iterations, but they do not decrease further even after a significant number of additional iterations~\footnote{We actually run the search with $k_{\max}   = 100,000$ iterations, but plot only the first 2,000 for clarity, since losses do not decrease significantly after 2,000 iterations.}, which suggests that the search algorithm is a very efficient practical approach for finding near optimal detector configurations.

\begin{figure}[h!]
\centering
\begin{tikzpicture}
\begin{axis}[
  width=\columnwidth,
  height=0.8\columnwidth,
  font=\small,
  xmax=2000,
  xlabel={Iterations $k_{\max}$},
  ylabel={False-positive cost and attack impact},
  grid=both,
  legend style={at={(0.97,0.425)},anchor=east},
  xtick={0, 500, 1000, 1500},
]
\addplot[dashed] table[x=iterations, y=strategic.cost, comment chars={\%}, col sep=comma] {software/multi_stage_strategies/plots/configuration_algorithms.csv};
\addlegendentry{cost of strategic};
\addplot[dash dot] table[x=iterations, y=uniform.cost, comment chars={\%}, col sep=comma] {software/multi_stage_strategies/plots/configuration_algorithms.csv};
\addlegendentry{cost of uniform};
\addplot[] table[x=iterations, y=strategic.impact, comment chars={\%}, col sep=comma] {software/multi_stage_strategies/plots/configuration_algorithms.csv};
\addlegendentry{impact from strat.};
\addplot[dotted] table[x=iterations, y=uniform.impact, comment chars={\%}, col sep=comma] {software/multi_stage_strategies/plots/configuration_algorithms.csv};
\addlegendentry{impact from unif.};
\end{axis}
\end{tikzpicture}
\caption{Defender's false-positive cost and attack impact resulting from uniform and strategic configurations output by the metaheuristic search algorithm.}
\label{fig:detector_configuration2}
\end{figure}
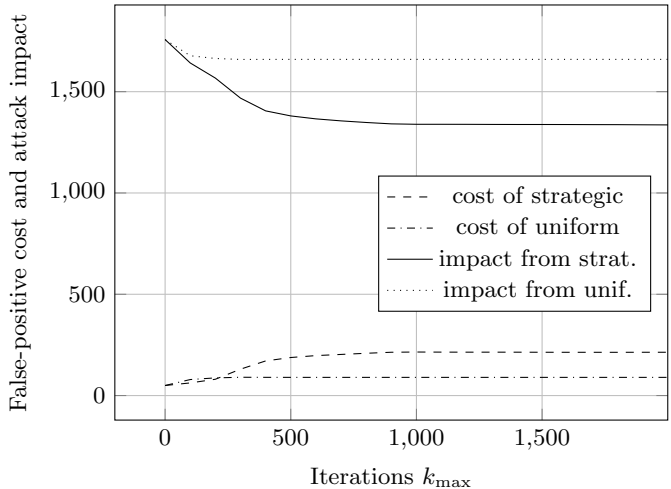

Figure~\ref{fig:detector_configuration2} shows the cost of false positives and the total impact of attacks with strategic and uniform detector configurations found by the search algorithm.
We observe that strategic detector configurations may result in slightly more false-positive errors, but they can significantly decrease the impact of attacks.

\section{Related Work}
\label{sec:related}

In this section, we briefly survey related work on the vulnerability of transportation networks, the optimal configuration of attack detectors, and game theory for security of cyber-physical systems.

\subsection{Vulnerability of Transportation Networks}

We first give a brief overview of the related work on the vulnerability of transportation networks.
%
A number of research efforts have studied the vulnerability of transportation networks to natural disasters and attacks.
However, to the best of our knowledge, our work is the first one to consider traffic-signal tampering attacks against general transportation networks.

In a closely related work, Reilly et al.\ consider the vulnerability of freeway control systems to attacks on the sensing and control infrastructure~\cite{reilly2014cybersecurity}.
They present an in-depth analysis on the takeover of a series of onramp-metering traffic lights using a methodology based on finite-horizon optimal control techniques and multi-objective optimization.

Prior work has studied the impact of other disruptive events as well.
Sullivan et al.\ study short-term disruptive events, such as partial flooding, and propose an approach that employs various link-based capacity-disruption values~\cite{sullivan2010identifying}.
The proposed approach can be used to identify and rank the most critical links and to quantify transportation network robustness (i.e., inverse vulnerability).
%
Jenelius and Mattson introduce an approach for systematically analyzing the robustness of road networks to disruptions affecting extended areas, such as floods and heavy snowfall~\cite{jenelius2012road}.
Their methodology is based on covering the area of interest with grids of uniformly shaped and sized cells, where each cell represents the extent of an event.
The authors apply their approach to the Swedish road network, and find that the impact of area-covering disruptions are largely determined by the internal, outbound, and inbound travel demands of the affected area itself.

In addition to assessing vulnerability, prior work has also considered the problems of identifying critical links and studying other aspects of vulnerabiltiy.
Scott et al.\ propose a comprehensive, system-wide approach for identifying critical links and evaluating network performance~\cite{scott2006network}.
Using three hypothetical networks, the authors demonstrate that their approach yields different highway planning solutions than traditional approaches, which rely on volume/capacity ratios to identify congested or critical links.
%
Jenelius proposes a methodology for vulnerability analysis of road networks and considers the impact of road-link closures~\cite{jenelius2010large}.
The author considers different aspects of vulnerability, and explores the dichotomy between system-wide efficiency and user equity.

Prior work has also considered game-theoretic models of attacks against transportation.
Alpcan and Buchegger investigate the resilience aspects of vehicular networks using a game-theoretic model, in which defensive measures are optimized with respect to threats posed by intentional attacks~\cite{alpcan2011security}.
The game is formulated in an abstract manner, based on centrality values computed by mapping the centrality values of the car communication network onto the road topology.
The authors consider multiple formulations based on varying assumptions on the players' information, and evaluate their models using numerical examples.
%
Bell introduces a two-player non-cooperative game between a network user, who seeks to minimize expected travel cost, and an adversary, who chooses link performance scenarios to maximize the travel cost~\cite{bell2000game,bell2008attacker}.
The Nash equilibrium of this game can be used to measure network performance when users are pessimistic and, hence, may be used for cautious network~design.
%
Wu and Amin study normal-form and sequential attacker-defender games over transportation networks to understand how a defender should prioritize its investment in securing a set of facilities~\cite{wu2018security}.

\subsection{Configuration of Detectors}


The problem of configuring the sensitivity of intrusion detection systems in the presence of strategic attackers has been studied in a variety of different ways in the academic literature~\cite{laszka2016optimal}.
For example, Alpcan and Basar study distributed intrusion detection in access control systems as a security game between an attacker and an IDS, using a model that captures the imperfect flow of information from the attacker to the IDS through a network~\cite{alpcan2003game,alpcan2004game}.
The authors investigate the existence of a unique Nash equilibrium and best-response strategies under specific cost functions, and analyze long-term interactions using repeated games and a dynamic model.
As another example, Dritsoula et al.\ consider the problem of setting a threshold for classifying an attacker into one of two categories, spammer and spy, based on its intrusion attempts~\cite{dritsoula2012computing}.
They give a characterization of the Nash equilibria in mixed strategies, and show that the equilibria can be computed in polynomial time.
More recently, Lis{\`y} et al.\ study randomized detection thresholds using a general model of adversarial classification, which can be applied to e-mail filtering, intrusion detection, steganalysis, etc.~\cite{lisy2014randomized}.
The authors analyze both Nash and Stackelberg equilibria based on the true-positive to false-positive curve of the classifier, and find that randomizing the detection threshold may force a strategic attacker to design less efficient attacks.
Finally, Zhu and Basar study the problem of optimal signature-based IDS configuration under resource constraints~\cite{zhu2011indices}.

The strategic configuration of the sensitivity of e-mail filtering against spear-phishing and other malicious e-mail is also closely related to the problem considered in this paper.
Laszka et al.\ study a single defender who has to protect multiple users against targeted and non-targeted malicious e-mail~\cite{laszka2015optimal}.
The authors focus on characterizing and computing optimal filtering thresholds, and they use numerical results to demonstrate that optimal thresholds can lead to substantially lower losses than na\"ive ones.
Zhao et al.\ study a variant of the previous model: they assume that the attacker can mount an arbitrary number of costly spear-phishing attacks in order to learn a secret, which is known only by a subset of the users~\cite{zhao2015initial,zhao2016optimal}.
They also focus on the computational aspects of finding optimal filtering thresholds; however, their variant of the model does not capture non-targeted malicious e-mails, such as~spam.

\subsection{Game Theory for Security of Cyber-Physical Systems}

\Aron{Check and revise if necessary!}
Beyond the configuration of detectors, prior efforts have also used game-theory to study a variety of other security problems in cyber-physical systems.
For instance,
Zhu and Basar introduce a game-theoretic framework for resilient control design and studying the trade-off between robustness, security, and resilience~\cite{zhu2015game}.
They employ a hybrid model, which integrates a discrete-time Markov model that captures the evolution of cyberstates with continuous-time dynamics that capture the underlying controlled physical process.
Backhaus et al.\ consider the problem of designing attack-resilient power grids and control systems~\cite{backhaus2013cyber}.
They use game theory to model the conflict between a cyber-physical intruder and a system operator, and use simulation results to assess design options.
Pawlick et al.\ consider Advanced Persistent Threats (APTs) against cloud-controlled cyber-physical systems and design a framework that specifies when a devices should trust commands from the cloud~\cite{pawlick2015flip}.
They model this scenario as a three player game between the cloud administrator, the attacker, and the device by combining the FlipIt game~\cite{van2013flipit,laszka2014flipthem} with a signaling game.
Li et al.\ consider jamming attacks against remote state estimation\cite{li2015jamming}.
In particular, they assume that a sensor and an estimate have to communicate over a wireless, formulate a game-theoretic model, and provide Nash equilibrium strategies.

\section{Conclusion}
\label{sec:concl}

As traffic-control devices in practice evolve into complex networks of smart devices, the risks posed to transportation networks by cyber-attacks increases.
Thus, it is imperative for traffic-network operators to be prepared to detect and mitigate attacks against traffic control.
To provide theoretical foundations for planning and implementing countermeasures, we introduced a game-theoretic model of cyber-attacks against traffic control.
Our security game model consists of three stages: 
defender configuring detectors, attacker mounting a tampering attack against traffic signals, and defender mitigating the attack. 
Since mitigation---in our model---means adapting the schedules of some traffic signals given the schedules of other signals (set by the adversary in the previous stage), it is equivalent to optimizing traffic control in a non-adversarial setting, which has been studied in prior work.
In light of this, we focused on the computational problem of finding optimal actions in the first two stages.
We showed that this is a computationally hard problem, which prompted us to propose efficient heuristic algorithms.

Using numerical results, we demonstrated that the proposed algorithms are practical.
In particular, we first showed that the greedy algorithm for attackers is close to optimal and computationally very efficient.
Second, we showed that the metaheuristic search algorithm for detector configuration is effective, and it can significantly decrease losses compared to non-strategic detector configuration.
We also introduced and studied a Gaussian-process based traffic-anomaly detector, which we showed to be very effective at detecting tampering attacks against traffic signals.

\section*{Acknowledgements}
This work was supported in part by the National Science Foundation under Grant IIS-1905558.

\bibliographystyle{elsarticle-num}
\bibliography{references}

\begin{thebibliography}{10}
\expandafter\ifx\csname url\endcsname\relax
  \def\url#1{\texttt{#1}}\fi
\expandafter\ifx\csname urlprefix\endcsname\relax\def\urlprefix{URL }\fi
\expandafter\ifx\csname href\endcsname\relax
  \def\href#1#2{#2} \def\path#1{#1}\fi

\bibitem{ghena2014green}
B.~Ghena, W.~Beyer, A.~Hillaker, J.~Pevarnek, J.~A. Halderman, Green lights
  forever: {A}nalyzing the security of traffic infrastructure, in: Proceedings
  of the 8th USENIX Workshop on Offensive Technologies (WOOT), 2014, pp. 1--10.

\bibitem{bernstein2007key}
S.~Bernstein, A.~Blankstein,
  \href{http://articles.latimes.com/2007/jan/09/local/me-trafficlights9}{Key
  signals targeted, officials say}, Los Angeles Times (January 2007).
\newline\urlprefix\url{http://articles.latimes.com/2007/jan/09/local/me-trafficlights9}

\bibitem{laszka2016vulnerability}
A.~Laszka, B.~Potteiger, Y.~Vorobeychik, S.~Amin, X.~Koutsoukos, Vulnerability
  of transportation networks to traffic-signal tampering, in: Proceedings of
  the 7th ACM/IEEE International Conference on Cyber-Physical Systems (ICCPS),
  2016, pp. 1--10.

\bibitem{daganzo1994cell}
C.~F. Daganzo, The cell transmission model: A dynamic representation of highway
  traffic consistent with the hydrodynamic theory, Transportation Research Part
  B: Methodological 28~(4) (1994) 269--287.

\bibitem{daganzo1995cell}
C.~F. Daganzo, The cell transmission model, part {II}: {N}etwork traffic,
  Transportation Research Part B: Methodological 29~(2) (1995) 79--93.

\bibitem{ziliaskopoulos2000linear}
A.~K. Ziliaskopoulos, A linear programming model for the single destination
  system optimum dynamic traffic assignment problem, Transportation science
  34~(1) (2000) 37--49.

\bibitem{laszka2019towards}
A.~Laszka, X.~Koutsoukos, Y.~Vorobeychik, Towards high-resolution multi-stage
  security games, in: Proactive and Dynamic Network Defense, Springer, 2019,
  pp. 139--161.

\bibitem{xie2010gaussian}
Y.~Xie, K.~Zhao, Y.~Sun, D.~Chen, Gaussian processes for short-term traffic
  volume forecasting, Transportation Research Record: Journal of the
  Transportation Research Board 2165~(1) (2010) 69--78.

\bibitem{chen2012decentralized}
J.~Chen, K.~H. Low, C.~K.-Y. Tan, A.~Oran, P.~Jaillet, J.~M. Dolan, G.~S.
  Sukhatme, Decentralized data fusion and active sensing with mobile sensors
  for modeling and predicting spatiotemporal traffic phenomena, in: Proceedings
  of the 28th Conference on Uncertainty in Artificial Intelligence (UAI), 2012,
  pp. 163--173.

\bibitem{rasmussen2006gaussian}
C.~E. Rasmussen, C.~K.~I. Williams, Gaussian Processes for Machine Learning,
  MIT Press, 2006.

\bibitem{blanchet2011spatial}
J.~Blanchet, A.~C. Davison, et~al., Spatial modeling of extreme snow depth,
  Annals of Applied Statistics 5~(3) (2011) 1699--1725.

\bibitem{kou2013sparse}
P.~Kou, F.~Gao, X.~Guan, Sparse online warped {G}aussian process for wind power
  probabilistic forecasting, Applied Energy 108 (2013) 410--428.

\bibitem{bazi2012improved}
Y.~Bazi, N.~Alajlan, F.~Melgani, Improved estimation of water chlorophyll
  concentration with semisupervised {G}aussian process regression, IEEE
  Transactions on Geoscience and Remote Sensing 50~(7) (2012) 2733--2743.

\bibitem{nevat2012location}
I.~Nevat, G.~W. Peters, I.~B. Collings, Location-aware cooperative spectrum
  sensing via gaussian processes, in: Proc. of the 13th Australian Comm. Theory
  Work. (AusCTW), IEEE, 2012, pp. 19--24.

\bibitem{krajzewicz2002sumo}
D.~Krajzewicz, G.~Hertkorn, C.~R{\"o}ssel, P.~Wagner, {SUMO} ({S}imulation of
  {U}rban {MO}bility): An open-source traffic simulation, in: Proceedings of
  the 4th Middle East Symposium on Simulation and Modelling (MESM), 2002, pp.
  183--187.

\bibitem{behrisch2011sumo}
M.~Behrisch, L.~Bieker, J.~Erdmann, D.~Krajzewicz, {SUMO} -- {S}imulation of
  {U}rban {MO}bility: {A}n overview, in: Proceedings of the 3rd International
  Conference on Advances in System Simulation (SIMUL), 2011, pp. 63--68.

\bibitem{gelman2013bayesian}
A.~Gelman, H.~S. Stern, J.~B. Carlin, D.~B. Dunson, A.~Vehtari, D.~B. Rubin,
  Bayesian data analysis, 3rd Edition, Chapman and Hall/CRC, 2013.

\bibitem{peng2012random}
W.~Peng, G.~Dong, K.~Yang, J.~Su, J.~Wu, A random road network model for
  mobility modeling in mobile delay-tolerant networks, in: Proceedings of the
  8th International Conference on Mobile Ad-hoc and Sensor Networks (MSN),
  IEEE, 2012, pp. 140--146.

\bibitem{peng2014random}
W.~Peng, G.~Dong, K.~Yang, J.~Su, A random road network model and its effects
  on topological characteristics of mobile delay-tolerant networks, IEEE
  Transactions on Mobile Computing 13~(12) (2014) 2706--2718.

\bibitem{reilly2014cybersecurity}
J.~Reilly, S.~Martin, M.~Payer, M.~Payer, On cybersecurity of freeway control
  systems: {A}nalysis of coordinated ramp metering attacks, Transportation
  Research, Part B.

\bibitem{sullivan2010identifying}
J.~Sullivan, D.~Novak, L.~Aultman-Hall, D.~M. Scott, Identifying critical road
  segments and measuring system-wide robustness in transportation networks with
  isolating links: A link-based capacity-reduction approach, Transportation
  Research Part A: Policy and Practice 44~(5) (2010) 323--336.

\bibitem{jenelius2012road}
E.~Jenelius, L.-G. Mattsson, Road network vulnerability analysis of
  area-covering disruptions: A grid-based approach with case study,
  Transportation research part A: policy and practice 46~(5) (2012) 746--760.

\bibitem{scott2006network}
D.~M. Scott, D.~C. Novak, L.~Aultman-Hall, F.~Guo, Network robustness index: a
  new method for identifying critical links and evaluating the performance of
  transportation networks, Journal of Transport Geography 14~(3) (2006)
  215--227.

\bibitem{jenelius2010large}
E.~Jenelius, Large-scale road network vulnerability analysis, Ph.D. thesis, KTH
  (2010).

\bibitem{alpcan2011security}
T.~Alpcan, S.~Buchegger, Security games for vehicular networks, IEEE
  Transactions on Mobile Computing 10~(2) (2011) 280--290.

\bibitem{bell2000game}
M.~G. Bell, A game theory approach to measuring the performance reliability of
  transport networks, Transportation Research Part B: Methodological 34~(6)
  (2000) 533--545.

\bibitem{bell2008attacker}
M.~G. Bell, U.~Kanturska, J.-D. Schm{\"o}cker, A.~Fonzone, Attacker--defender
  models and road network vulnerability, Philosophical Transactions of the
  Royal Society of London A: Mathematical, Physical and Engineering Sciences
  366~(1872) (2008) 1893--1906.

\bibitem{wu2018security}
M.~Wu, S.~Amin, Security of transportation networks: Modeling attacker-defender
  interaction, arXiv preprint arXiv:1804.00391,
  \url{https://arxiv.org/abs/1804.00391v1}.

\bibitem{laszka2016optimal}
A.~Laszka, W.~Abbas, S.~S. Sastry, Y.~Vorobeychik, X.~Koutsoukos, Optimal
  thresholds for intrusion detection systems, in: Proceedings of the 3rd Annual
  Symposium and Bootcamp on the Science of Security (HotSoS), 2016, pp. 72--81.

\bibitem{alpcan2003game}
T.~Alpcan, T.~Basar, A game theoretic approach to decision and analysis in
  network intrusion detection, in: Proceedings of the 42nd IEEE Conference on
  Decision and Control (CDC), Vol.~3, IEEE, 2003, pp. 2595--2600.

\bibitem{alpcan2004game}
T.~Alpcan, T.~Ba{\c{s}}ar, A game theoretic analysis of intrusion detection in
  access control systems, in: Proceedings of the 43rd IEEE Conference on
  Decision and Control (CDC), Vol.~2, IEEE, 2004, pp. 1568--1573.

\bibitem{dritsoula2012computing}
L.~Dritsoula, P.~Loiseau, J.~Musacchio, Computing the {N}ash equilibria of
  intruder classification games, in: Proceedings of the 3rd International
  Conference on Decision and Game Theory for Security (GameSec), Springer,
  2012, pp. 78--97.

\bibitem{lisy2014randomized}
V.~Lis{\`y}, R.~Kessl, T.~Pevn{\`y}, Randomized operating point selection in
  adversarial classification, in: Proceedings of the 2014 European Conference
  on Machine Learning and Principles and Practice of Knowledge Discovery in
  Databases (ECML PKDD), Part II, Springer, 2014, pp. 240--255.

\bibitem{zhu2011indices}
Q.~Zhu, T.~Ba{\c{s}}ar, Indices of power in optimal {IDS} default
  configuration: {T}heory and examples, in: Proceedings of the 2nd
  International Conference on Decision and Game Theory for Security (GameSec),
  Springer, 2011, pp. 7--21.

\bibitem{laszka2015optimal}
A.~Laszka, Y.~Vorobeychik, X.~Koutsoukos, Optimal personalized filtering
  against spear-phishing attacks, in: Proceedings of the 29th AAAI Conference
  on Artificial Intelligence (AAAI), 2015, pp. 958--964.

\bibitem{zhao2015initial}
M.~Zhao, B.~An, C.~Kiekintveld, An initial study on personalized filtering
  thresholds in defending sequential spear phishing attacks, in: Proceedings of
  the 2015 IJCAI Workshop on Behavioral, Economic and Computational
  Intelligence for Security, 2015, pp. 1--9.

\bibitem{zhao2016optimal}
M.~Zhao, B.~An, C.~Kiekintveld, Optimizing personalized email filtering
  thresholds to mitigate sequential spear phishing attacks, in: Proceedings of
  the 30th AAAI Conference on Artificial Intelligence (AAAI), 2016, pp.
  658--664.

\bibitem{zhu2015game}
Q.~Zhu, T.~Basar, Game-theoretic methods for robustness, security, and
  resilience of cyberphysical control systems: games-in-games principle for
  optimal cross-layer resilient control systems, IEEE Control Systems Magazine
  35~(1) (2015) 46--65.

\bibitem{backhaus2013cyber}
S.~Backhaus, R.~Bent, J.~Bono, R.~Lee, B.~Tracey, D.~Wolpert, D.~Xie,
  Y.~Yildiz, Cyber-physical security: A game theory model of humans interacting
  over control systems, IEEE Transactions on Smart Grid 4~(4) (2013)
  2320--2327.

\bibitem{pawlick2015flip}
J.~Pawlick, S.~Farhang, Q.~Zhu, Flip the cloud: Cyber-physical signaling games
  in the presence of advanced persistent threats, in: Proceedings of the 6th
  Conference on Decision and Game Theory for Security (GameSec), Springer,
  2015, pp. 289--308.

\bibitem{van2013flipit}
M.~Van~Dijk, A.~Juels, A.~Oprea, R.~L. Rivest, {FlipIt}: The game of
  “stealthy takeover”, Journal of Cryptology 26~(4) (2013) 655--713.

\bibitem{laszka2014flipthem}
A.~Laszka, G.~Horvath, M.~Felegyhazi, L.~Buttyan, {FlipThem}: Modeling targeted
  attacks with {FlipIt} for multiple resources, in: Proceedings of the 5th
  Conference on Decision and Game Theory for Security (GameSec), 2014, pp.
  175--194.

\bibitem{li2015jamming}
Y.~Li, L.~Shi, P.~Cheng, J.~Chen, D.~E. Quevedo, Jamming attacks on remote
  state estimation in cyber-physical systems: A game-theoretic approach, IEEE
  Transactions on Automatic Control 60~(10) (2015) 2831--2836.

\end{thebibliography}

\listoftodos

\end{document}